\documentclass[sigconf, screen, nonacm]{acmart}
\usepackage{subfiles}
\usepackage{style/customize}
\usepackage{physics}
\usepackage{comment}

\AtBeginDocument{%
  }

\setcopyright{none} 
\copyrightyear{2024}
\acmYear{2024}
\acmDOI{XXXXXXX.XXXXXXX}

\begin{document}

%\title{The Name of the Title Is Hope}
\title{High-Performance and Scalable Fault-Tolerant Quantum Computation with Lattice Surgery on a 2.5D Architecture}

%
% The "author" command and its associated commands are used to define
% the authors and their affiliations.
% Of note is the shared affiliation of the first two authors, and the
% "authornote" and "authornotemark" commands
% used to denote shared contribution to the research.
\author{Yosuke Ueno}
%\orcid{0000-0002-0402-9914}
\affiliation{%
  \institution{RIKEN}
  \city{Saitama}
  \country{Japan}
}
\email{yosuke.ueno@riken.jp}

\author{Taku Saito}
\affiliation{%
  \institution{The University of Tokyo}
  \city{Tokyo}
  \country{Japan}
}
\email{saito-taku808@g.ecc.u-tokyo.ac.jp}

\author{Teruo Tanimoto}
\affiliation{%
  \institution{Kyushu University}
  \city{Fukuoka}
  \country{Japan}
}
\email{tteruo@kyudai.jp}

\author{Yasunari Suzuki}
\affiliation{%
  \institution{NTT}
  \city{Tokyo}
  \country{Japan}
}
\email{yasunari.suzuki@ntt.com}

\author{Yutaka Tabuchi}
\affiliation{%
  \institution{RIKEN}
  \city{Saitama}
  \country{Japan}
}
\email{yutaka.tabuchi@riken.jp}

\author{Shuhei Tamate}
\affiliation{%
  \institution{RIKEN}
  \city{Saitama}
  \country{Japan}
}
\email{shuhei.tamate@riken.jp}

\author{Hiroshi Nakamura}
\affiliation{%
  \institution{The University of Tokyo}
  \city{Tokyo}
  \country{Japan}
}
\email{nakamura@hal.ipc.i.u-tokyo.ac.jp}

\renewcommand{\shortauthors}{Ueno et al.}

\begin{abstract}
Due to the high error rate of a qubit, detecting and correcting errors on it is essential for fault-tolerant quantum computing (FTQC).
Among several FTQC techniques, lattice surgery (LS) using surface code (SC) is currently promising.
To demonstrate practical quantum advantage as early as possible, it is indispensable to propose a high-performance and low-overhead FTQC architecture specialized for a given FTQC scheme based on detailed analysis.

In this study, we first categorize the factors, or \textit{hazards}, that degrade LS-based FTQC performance and propose a performance evaluation methodology to decompose the impact of each hazard, inspired by the CPI stack. 
We propose the \textit{Bypass architecture} based on the bottleneck analysis using the proposed evaluation methodology.
The proposed Bypass architecture is a 2.5-dimensional architecture consisting of dense and sparse qubit layers and successfully eliminates the bottleneck to achieve high-performance and scalable LS-based FTQC.
We evaluate the proposed architecture with a circuit-level stabilizer simulator and a cycle-accurate LS simulator with practical quantum phase estimation problems.
The results show that the Bypass architecture improves the fidelity of FTQC and achieves both a 1.73$\times$ speedup and a 17\% reduction in classical/quantum hardware resources over a conventional 2D architecture.
\end{abstract}

%\begin{CCSXML}
%<ccs2012>
%  <concept_id>00000000.0000000.0000000</concept_id>
%  <concept_desc>Do Not Use This Code, Generate the Correct Terms for Your Paper</concept_desc>
%  <concept_significance>500</concept_significance>
%</ccs2012>
%\end{CCSXML}

%\ccsdesc[500]{Do Not Use This Code~Generate the Correct Terms for Your Paper}
%\ccsdesc[300]{Do Not Use This Code~Generate the Correct Terms for Your Paper}
%\ccsdesc{Do Not Use This Code~Generate the Correct Terms for Your Paper}
%\ccsdesc[100]{Do Not Use This Code~Generate the Correct Terms for Your Paper}

%\keywords{Do, Not, Us, This, Code, Put, the, Correct, Terms, for, Your, Paper}


\maketitle
%\vspace{-10mm}

\section{Introduction\label{sec:introduction}}

The inherent noise of quantum computers (QCs) poses a significant obstacle to practical quantum algorithms.
To implement fault-tolerant quantum computation (FTQC) effectively, it is crucial to select an appropriate quantum error correction (QEC) code.
The most promising approach currently is the surface code (SC) combined with lattice surgery (LS). 
SCs leverage a two-dimensional (2D) grid of qubits, where logical qubits are encoded using multiple physical qubits. 
LS facilitates operations on SC-based logical qubits by dynamically modifying their boundaries, effectively ``expanding'' and ``merging'' them~\cite{horsman2012surface}.
In LS-based FTQC, logical qubits are categorized into \textit{data cells}, which store logical states, and \textit{ancillary cells}, which facilitate logical operations on data cells, as detailed in Sec.~\ref{subsec:qubit_plane}.
Due to the constraints of nearest-neighbor connections on physical qubits, particularly in solid-state qubits such as superconducting qubits, ancillary cells are necessary for implementing LS operations. 
Increasing the ratio of data cells $R_{data}$ is essential to reduce the number of extra qubits required.

In general, FTQC performance involves a tradeoff with the hardware resources required.
The performance of LS-based FTQC is determined by various factors, such as the magic state generation rate, paths for LS operations, and the latency and throughput of classical computations for the error-decoding process.
In addition, these factors are complicatedly affected by the $R_{data}$, which makes optimization more difficult.
Thus, a methodology to accurately assess FTQC performance while considering these factors is necessary for achieving high-performance FTQC with minimal resources.

In this work, we address these challenges by introducing a performance analysis tool called the \textit{Code Beats Per Instruction (CBPI) stack}, inspired by the Cycles per Instruction (CPI) stack in classical computing. 
Based on the insights gained from the CBPI analysis, which identified LS path conflicts and long LS paths as significant bottlenecks for scalable FTQC, we propose a novel architecture called the \textit{Bypass Architecture} to address these issues.

The CBPI stack breaks down the impact of various factors on execution time, referred to as \textit{FTQC hazards}, providing a clear understanding of performance bottlenecks, as detailed in Sec.~\ref{sec:performance_metric}. 
Using this tool, we analyzed a quantum phase estimation (QPE) problem under various $R_{data}$ configurations and obtained trade-offs as shown in Fig.~\ref{fig:CBPI_preliminary}. 
As indicated by the ``Path'' in the figure, the arrangement with $R_{data} =50\%$ hindered the simultaneous execution of independent operations due to LS path conflicts.
In addition, the roundabout paths required for LS operations in the $50\%$ arrangement decrease performance, as indicated by ``Decoding'' because these longer paths impose heavier loads on decoders.

\begin{figure}[tb]
    \centering
    \includegraphics[width=\linewidth]{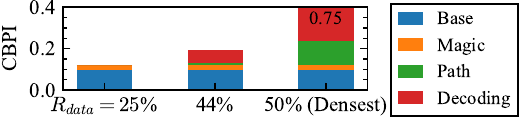}
    \caption{FTQC performance evaluation using CBPI stack for a QPE program. Section~\ref{subsec:setup_LS_simulation} details the experimental setup.}
    \label{fig:CBPI_preliminary}
\end{figure}

The impact of longer LS paths poses a significant challenge not only to execution time but also to the scalability of FTQC systems.
The logical error rate (LER) per operation is approximately proportional to the size of the syndrome graph used for error decoding, meaning that LS operations with longer paths have a higher LER.
To mitigate the impact of high error rate operations, it is necessary to either increase the redundancy of QEC codes or employ more precise decoders.
However, these approaches increase the required quantum and classical resources, worsening the scalability of FTQC systems.
As the scale of FTQC programs grows, the impact of these issues becomes more pronounced because longer LS paths are required.
Thus, reducing the LS path length while maintaining a high $R_{data}$ is essential for scalable FTQC systems.

Wiring length is also a well-known challenge in designing microprocessors for classical computers, and solutions involving layered interconnects have been extensively studied.
However, in LS-based FTQC, leveraging a wiring layer to reduce the length of LS paths has not yet been explored.
This work proposes the Bypass architecture to achieve high-performance and scalable LS-based FTQC with moderate resource overhead by suppressing the length of LS paths.
The Bypass architecture features a 2.5D qubit layout composed of a regular layer and a sparse layer, as shown in Fig.~\ref{fig:research_concept}. 
Our architecture reduces the impact of the decoding process by providing effectively shorter LS paths through the sparse layer and resolves path conflicts by increasing the number of possible paths for LS.

We evaluate the Bypass architecture for essential FTQC subroutines for practical QPE tasks with a cycle-accurate LS simulator and show that it achieves a speedup over the baseline 2D architecture and other 3D architectures with two-qubit layers.
In addition, our 2.5D architecture requires fewer classical and quantum hardware resources than the other architectures, thereby enhancing the scalability of the FTQC system.

Our contributions are summarized as follows.
\begin{enumerate}
\item We summarize hazards affecting LS-based FTQC performance and propose a performance analysis methodology that assesses the impact of each hazard.  (Sec.~\ref{sec:instruction_scheduling} and \ref{sec:performance_metric})
\item We propose the Bypass architecture with a 2.5D qubit layout for high-performance and scalable FTQC. (Sec.~\ref{sec:proposal})
\item Our evaluation shows that the Bypass architecture achieves both a 1.73$\times$ speedup and 17\% reduction in hardware resources over the 2D architecture for a practical QPE program in the moderate resources case. (Sec.~\ref{sec:QEC_performance} and \ref{sec:CBPI_evaluation})
\end{enumerate}

\begin{figure}[tb]
    \centering
    \includegraphics[width=\linewidth]{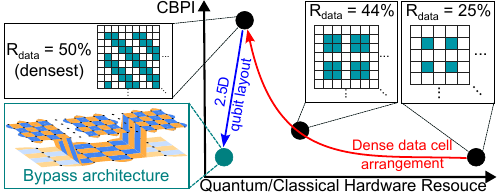}
    \caption{Our research goal.}
    \label{fig:research_concept}
\end{figure}

\section{Background on LS-based FTQC\label{sec:background}}
\subsection{Quantum error correction with SCs\label{subsec:surface_code}}
\begin{figure*}[tb]
    \centering
    \includegraphics[width=\linewidth]{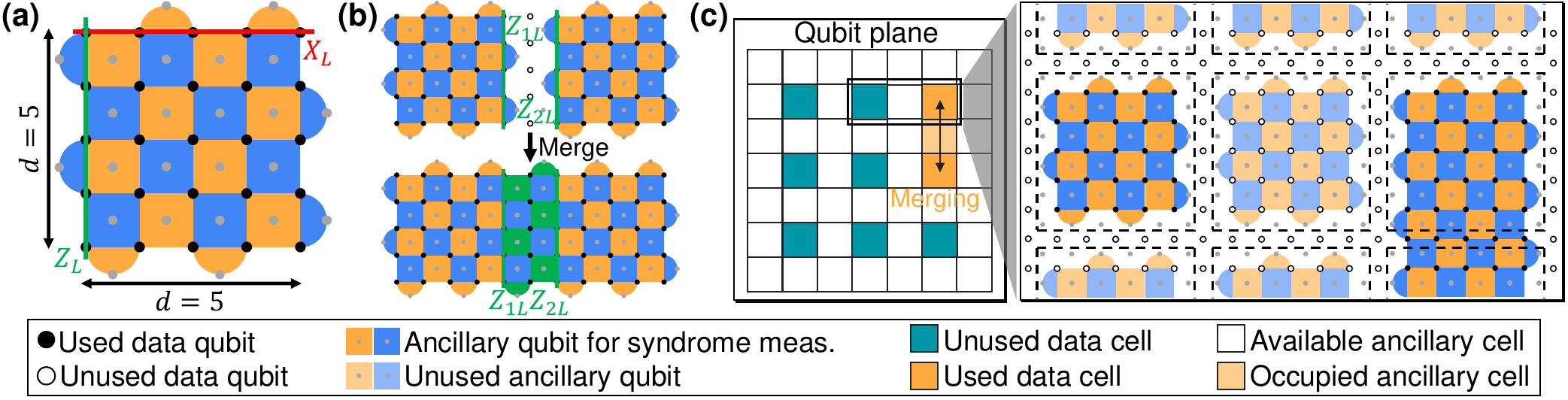}
    \caption{(a) Stabilizer-level picture of SC ($d=5$). 
    (b) Merge operation for \measzz instruction (Stabilizer-level view). (c) Cell- and stabilizer-level views of qubit plane during \measxx instruction.}
    \label{fig:surface_code_d3}
\end{figure*}

The SC is one of the most promising QEC codes, which can be implemented on a 2D qubit plane~\cite{bravyi1998quantum,kitaev1997quantum}. 
Figure~\ref{fig:surface_code_d3}\,(a) shows a schematic picture of a rotated SC with code distance $d = 5$.
SC consists of two types of physical qubits: data and ancillary qubits. 
Data qubits represent a logical qubit, while ancillary qubits are utilized to check the parity of errors on the neighboring data qubits. 
This parity check and its binary outcome are called a stabilizer measurement and a syndrome value, respectively. 
To deal with both bit-flip (Pauli-$X$) and phase-flip (Pauli-$Z$) errors, QEC needs two types of stabilizer measurements, \textit{i.e.}, $X$- and $Z$-stabilizer measurements. 
$X$- and $Z$-stabilizer measurements can detect the Pauli-$Z$ and Pauli-$X$ errors on the neighboring data qubits, respectively.

If we assume a noise model where Pauli errors occur probabilistically on data qubits, the estimation of the most likely Pauli errors can be reduced to the minimum-weight perfect matching (MWPM) problem on the decoding graph. 
In this graph, nodes and edges correspond to the syndrome values and the data qubits, respectively~\cite{fowler2012surface}.
Even when ancillary qubits also suffer from noise, we can reliably estimate errors by extending the MWPM problem to a 3D decoding graph by considering stacked 2D snapshots of $d$ syndrome measurements. 
In this paper, we refer to the time taken for one syndrome measurement as a \textit{code cycle} and that for $d$ syndrome measurements, \textit{i.e.}, $d$ code cycles, as a \textit{code beat}.

\subsection{LS and gate teleportation with magic states\label{subsec:bg_lattice_surgery}}

For universal FTQC, we need to perform a universal gate set on encoded logical qubits in a fault-tolerant manner.
The standard logical operations set is summarized as follows. 
\begin{itemize}
    \item Initialization of a logical qubit in a $Z$ or $X$ basis.
    \item Destructive measurement of a logical qubit in $Z$ or $X$ basis.
    \item Single-qubit operations, such as Hadamard and phase gates, and $T$ gates with magic states.
    \item Multi-qubit operations via multi-body Pauli measurements. 
\end{itemize}
Initialization and destructive measurement of a logical qubit can be achieved straightforwardly. 
In addition, Hadamard and phase gates are performed by expanding and shrinking SCs through additional syndrome measurement.

The LS technique implements multi-qubit Pauli measurements on SC-based logical qubits using only neighboring physical qubit operations.
One of the minimum examples of LS, a logical Pauli-$ZZ$ measurement on two logical qubits, is shown in Fig.~\ref{fig:surface_code_d3}\,(b). 
As shown in the figure, 1) we initialize all the sandwiched physical qubits to physical $\ket{+}$ states, 2) two SCs are merged by performing another set of stabilizer measurements and repeating them for one code beat, and 3) split it into two planes by performing the original stabilizer measurements and all the sandwiched physical qubits are measured in the Pauli-$X$ basis.
This merge operation with the smooth boundaries implements a logical Pauli-$ZZ$ measurement on the two logical qubits. 
The outcome of the logical Pauli measurement is calculated from the parity of the outcomes of Pauli-$X$ stabilizer measurements in the first cycle of LS. 
The Pauli-$XX$ measurement can also be performed similarly by merging the rough boundaries of logical qubits.

The logical $T$ gate is indirectly performed with the gate-teleportation technique by consuming a logical qubit prepared in the magic state $\ket{M} = (\ket{0} + e^{i\pi/4} \ket{1})/\sqrt{2}$~\cite{bravyi_magic_state}. 
While the direct preparation of high-fidelity $\ket{M}$ in the logical space is difficult, the MSD protocol constructs a clean magic state from several noisy magic states~\cite{bravyi_magic_state}. 
The area for MSD is called the \textit{magic-state factory}, or \textit{factory}.

Factories introduce space and time overhead to FTQC.
For example, a typical factory implementation with 15-to-1 MSD protocol~\cite{bravyi_magic_state} requires a space of 24 SC cells and five repetitions of eight-qubit logical Pauli measurements to prepare a clean magic state.
Therefore, factories are considered a major bottleneck in large-scale FTQC~\cite{babbush2018encoding,gidney2021rsa}. 
However, many recent theoretical studies have proposed more efficient MSD protocols and efficient implementations of factories to reduce their costs~\cite{gidney2019efficient,litinski2019magic,tan2024sat,itogawa2024even,hirano2024leveraging,gidney2024magic}.

\subsection{Qubit plane for LS-based FTQC\label{subsec:qubit_plane}}
We suppose that qubits are integrated on a 2D plane, called the \textit{qubit plane}, in the baseline FTQC architecture to perform the essential logical operations in the previous subsection.
The qubit plane is segmented into multiple distance-$d$ SC cells, as represented by the dashed squares on the right side of Fig.~\ref{fig:surface_code_d3}\,(c).
Note that each cell contains $d^2$ data qubits and $(d+1)^2$ ancillary qubits, with all data qubits and $d^2-1$ ancillary qubits used for distance-$d$ SC.

Each cell on a qubit plane is divided into two roles as follows.
Throughout the FTQC process, certain cells are designated for storing single-qubit information as logical qubits, which we call \textit{data cells}. 
On the other hand, additional cells, termed \textit{ancillary cells}, serve as functional areas for executing logical operations on data cells. 
In addition, there is a space between each cell, consisting of $d$ data qubits, which we call the \textit{cell gap} or simply the \textit{gap}.
The left side of Fig.~\ref{fig:surface_code_d3}\,(c) shows a cell-level view of a qubit plane, which depicts data and ancillary cells with green and white, respectively.

During a logical operation, data cells and one or several ancillary cells may be occupied corresponding to the operation, as shown on the right side of the figure. 
Given that SCs feature two types of boundaries, each cell is aligned in two orientations: the top and bottom as $X$ boundaries and the remaining sides as $Z$ boundaries, or vice versa.
For simplicity, this paper uses only the former orientation.
When operations that alter the boundary configuration, such as a logical Hadamard gate, are performed, we compensate for these changes by using SC rotation operations~\cite{litinski2019game}.

\subsection{LS instruction set\label{subsec:instruction_set}}

\begin{table}[tb]
\caption{LS instruction set used in this paper.}
\label{tab:instruction_set}
\scriptsize
\begin{tabular}{|l|l|l|l|} \hline 
Operation                                                                        & Duration     & Operands & Effect                                                                                                         \\ \hline \hline
\begin{tabular}[c]{@{}l@{}}\texttt{INIT\_Z}\\ (\texttt{INIT\_X})\end{tabular}    & 1 code cycle & $a_0$& \begin{tabular}[c]{@{}l@{}}Initialize ancillary cell $a_0$\\ in logical $\ket{0}$ ($\ket{+}$) state.\end{tabular}  \\ \hline
\texttt{OP\_H}                                                                   & 3 code beats & $d_0$& \begin{tabular}[c]{@{}l@{}}Perform Hadamard gate on data cell $d_0$.\end{tabular}                        \\ \hline
\texttt{OP\_S}                                                                   & 2 code beats & $d_0$& \begin{tabular}[c]{@{}l@{}}Perform phase gate on data cell $d_0$~\cite{beverland2022assessing}.\end{tabular} \\ \hline
\begin{tabular}[c]{@{}l@{}}\texttt{MEAS\_Z}\\ (\texttt{MEAS\_X})\end{tabular}    & 1 code cycle & $d_0, c$& \begin{tabular}[c]{@{}l@{}}Destructively measure data cell $d_0$ \\  in $Z$ ($X$) basis.\end{tabular}            \\ \hline
\begin{tabular}[c]{@{}l@{}}\texttt{MEAS\_ZZ}\\ (\texttt{MEAS\_XX})\end{tabular}  & 1 code beat  & $d_0, d_1, c$& \begin{tabular}[c]{@{}l@{}}Perform $ZZ$ ($XX$) measurement\\on two data cells $d_0$ and $d_1$.\end{tabular} \\\hline
\end{tabular}
\end{table}

\begin{figure}[tb]
    \centering
    \includegraphics[width=\linewidth]{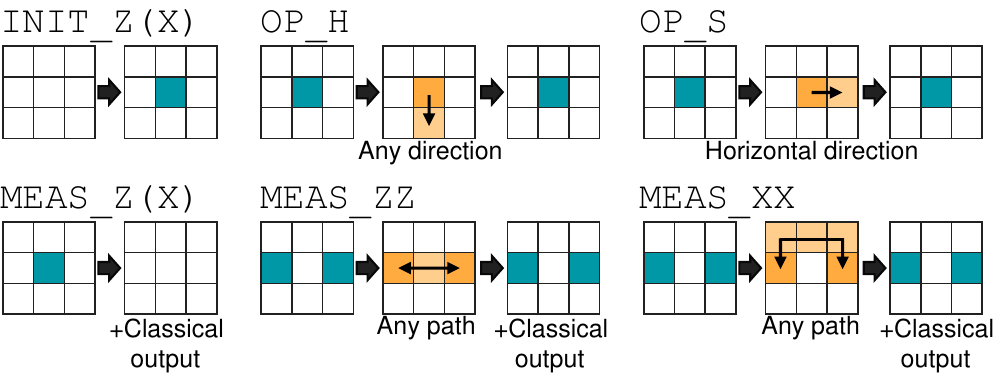}
    \caption{Procedure of each instruction in Table~\ref{tab:instruction_set}.}
    \label{fig:instruction_set}
\end{figure}

\begin{figure*}[tb]
    \centering
    \includegraphics[width=\linewidth]{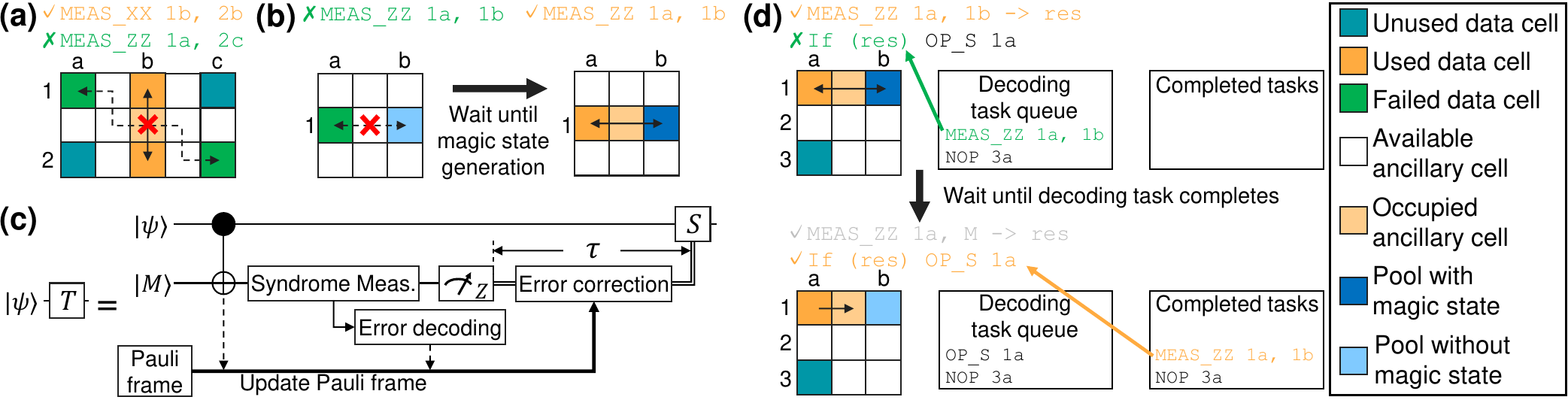}
    \caption{(a) Path hazard caused by two LS operations. 
    (b) Magic hazard.
    (c) A gate-teleportation circuit to perform $T$ gate with a magic state $\ket{M}$ and its error decoding scheme with Pauli frame\cite{knill2005quantum}. 
    (d) Decoding hazard.}
    \label{fig:hazards}
\end{figure*}

To execute a given quantum program with LS-based FTQC, we need to translate it into a sequence of LS-operable instructions.
In this paper, we refer to a set of such instructions as \textit{LS instruction set} or simply \textit{instruction set}.
An LS instruction set contains the following types of instructions to support universal quantum computation: logical state initialization of a cell, magic-state generation at magic-state factories, $S$, $H$, one-body Pauli measurements, and multi-body Pauli measurements with merge-and-split of data cells.
Note that it is unnecessary to include logical Pauli operations in the instruction set because they can be implemented using the Pauli frame~\cite{knill2005quantum} without actual manipulation of qubits.

In this paper, we focus on the succinct instruction set summarized in Table~\ref{tab:instruction_set}. 
The $d_i$ ($a_i$) in the Operands column represents the $i$-th data (ancillary) cell, and the $c$ represents the 1-bit classical register to store a logical measurement outcome. 
Figure~\ref{fig:instruction_set} visualizes each instruction of the set on a qubit plane. 
As explained in Sec.~\ref{subsec:bg_lattice_surgery}, one-qubit instructions, such as \oph and \opsgate, are performed by expanding and shrinking a data cell.
In addition, the \measzz and \measxx are executed by merging and splitting the appropriate boundaries of two data cells with a path of ancillary cells.

As shown in the figure, the choice of direction for expanding the data cell (path between the two data cells) during \oph and \opsgate (\measzz and \measxx) instructions is arbitrary and is not determined by the instructions themselves.
To perform LS-based FTQC, we need to appropriately select the directions and the paths to map the LS instructions sequence onto a qubit plane.
As detailed in Sec.~\ref{sec:instruction_scheduling}, we assume that these directions and paths are determined dynamically during the instruction scheduling.

Note that the instruction set does not include any operations for magic-state generation to be separated from the implementation of the factory, which has many variations~\cite{litinski2019game,litinski2019magic}.
For simplicity, we assume each magic state generated in a factory is supplied to a corresponding cell, which we call \textit{magic-state pool}.

This paper assumes all FTQC programs are performed with the instruction set.
However, note that our discussion is not limited to the specific instruction set and applies to any instruction set of FTQC, as long as it is built based on topological stabilizer codes.

\section{LS instruction scheduling and hazards\label{sec:instruction_scheduling}}

\subsection{LS instruction scheduling}

As shown in Fig.~\ref{fig:instruction_set}, LS instructions are executed while occupying specific cells on the qubit plane.
Therefore, scheduling LS instructions involves mapping a given 1D sequence of LS instructions onto the qubit plane appropriately. 
Note that this process includes determining the direction for expanding the data cell for \oph and \opsgate and the path between the two data cells for \measzz and \measxx.

For high-performance FTQC, it is desirable to utilize the instruction-level parallelism of the program and execute as many instructions simultaneously as possible.
In this study, we assume that LS instructions can be executed out-of-order as long as there are no data dependencies.
In such a situation, ideally, the FTQC performance would be determined by the maximum depth of data dependencies in the program.
However, in reality, various factors, which we call hazards, hinder the execution of LS instructions and degrade FTQC performance, as summarized in Sec.~\ref{subsec:hazards}.
In addition, these hazards introduce tradeoffs between execution time and resources.

To simplify the discussion, we assume that the LS instruction scheduling follows a greedy policy, where any executable instructions are always executed within a code beat, and the shortest available path for the given \measzz and \measxx instructions is chosen.
Note that optimal scheduling is known to be an NP-hard problem~\cite{herr2017optimization,molavi2023compilation}.

\subsection{Hazards on LS-based FTQC\label{subsec:hazards}}

\subsubsection{Path conflict}
To execute an instruction on the qubit plane, appropriate ancillary cell(s) must be available for the instruction and its operand data cell(s).
Otherwise, the instruction will not be executed, which is called ``path hazard''.
Figure~\ref{fig:hazards}\,(a) shows an example where the execution of the \measzz is blocked due to the absence of an appropriate ancillary path between the data cells because the path is occupied by the \measxx.
If other instructions occupy ancillary cells required for the target instruction, we need to wait for a certain number of code beats until the required cells are freed to execute the target instruction.
We refer to the additional code beats caused by path hazards as ``path penalty''.

Possible approaches to decrease penalties on FTQC program execution are increasing the ratio of ancillary to data cells~\cite{lao2018mapping} on a qubit plane and building an LS architecture with multiple qubit layers~\cite{viszlai2023architecture}.
However, both approaches require an increased number of physical qubits.

\subsubsection{Shortage of magic states}
Magic states are consumed by non-Clifford gates and are generated at magic-state factories at regular intervals of several code beats.
For example, the factory proposed in Ref.~\cite{litinski2019game} requires 15 code beats to produce a magic state with a sufficient LER.
As shown in Fig.~\ref{fig:hazards}\,(b), if all magic states are consumed when an instruction requiring a magic state is to be executed, it will not be executed, which we call a ``magic hazard''.
When magic hazards occur, we must delay the instructions with magic states until factories generate new magic states.
Additional code beats due to magic hazards are called ``magic penalty''.

We can decrease the penalty by increasing the number of factories. 
However, since each factory demands many physical qubits, this approach introduces hardware overheads.

\subsubsection{Decoding process\label{subsubsec:decoding_process}}
When performing a logical $T$ gate using the gate teleportation technique with a magic state, we must decide whether to apply an $S$ gate based on the logical measurement of a logical qubit prepared as a magic state, as shown in Fig.~\ref{fig:hazards}\,(c).
For the conditional branch, the logical measurement result must be reliable, \textit{i.e.}, all error-decoding tasks associated with the logical measurement must be completed before the branch. 
If any decoding tasks remain when performing gate teleportation, the logical operation controlled by the logical measurement result will not be executed, which we call a ``decoding hazard''.
Figure~\ref{fig:hazards}\,(d) illustrates an example of decoding hazard, where \opsgate is controlled by the result of \measzz with a magic state.
In this situation, the \opsgate is executed after the decoding task for the \measzz is completed.
Decoding is also necessary to protect data cells not executing instructions, represented as \nopgate in the figure.
We refer to the additional code beats caused by decoding hazards as a ``decoding penalty''.

A naive approach is to prepare ample classical computational resources for decoding; a decoder with sufficiently low latency and high throughput, capable of complete online decoding for any decoding task, can completely mitigate the hazards ~\cite{terhal2015quantum,skoric2023parallel,battistel2023real}. 
However, a computationally expensive decoder may not be desirable for the scalability of FTQC systems. 
Specifically, in scalable superconducting FTQC systems, where the decoder must be placed inside a cryogenic environment~\cite{tannu2017cryogenic,holmes2020nisq,ueno2021qecool,ueno2022qulatis,ueno2022neo,byun2022xqsim}, the decoder must therefore be resource-efficient. 

Another approach is to simplify decoding tasks during the execution of FTQC programs, making them manageable even with a resource-limited decoder.
In general, the difficulty of decoding tasks increases proportionally with the size of the syndrome graph for decoding; LS operations with longer paths impose more challenging decoding tasks on decoders. 
Thus, executing FTQC programs with shorter LS paths mitigates the decoding penalty.

\section{Performance analytical methodology\label{sec:performance_metric}}
To conduct a bottleneck analysis and appropriately improve LS-based FTQC architecture, we propose a metric named ``code beats per instruction (CBPI)'' and a performance analytical methodology named ``CBPI stack'', inspired by the concept of CPI and the CPI stack.
The CPI stack is a performance analysis methodology for processors that decomposes the overall CPI into distinct categories based on sources of performance loss, such as cache misses and branch mispredictions. 
This breakdown allows architects to quantify the impact of each factor on processor efficiency, thereby identifying potential areas for optimization.

The CBPI represents the average number of code beats required to execute a single instruction, enabling the estimation of latency for each instruction and the overall execution time of an FTQC program. 
CBPI is calculated for a specific benchmark program.
Its value varies based on the architecture configuration, including the data cell arrangement, the number and generation rate of magic-state factories, and the QEC decoder resources.

The CBPI stack is a methodology for visualizing performance bottlenecks, inspired by the CPI stack, as shown in Fig.~\ref{fig:CBPI_preliminary}.
It allows us to identify and prioritize areas for improving FTQC architectures by breaking down the impact of each penalty on CBPI. 
The size of each part in the CBPI stack is proportional to its impact on the total CBPI, ensuring that the sum of all parts equals the total CBPI.

To display the CBPI stack, CBPI is calculated for scenarios where each hazard is ignored, and the residuals are then stacked.
First, we begin with the `Base', which represents the ideal CBPI without any hazards. 
The Base CBPI is calculated through simulation with an infinite magic state generation rate by factories, simultaneous execution of instructions despite intersecting paths, and fully real-time decoding.
Then, the impact of each hazard is considered individually to calculate the differences in CBPI, which are displayed as a stacked graph to form the CBPI stack.

Note that calculating the CPI stack accurately is challenging when multiple hazards may occur concurrently. 
This issue similarly applies to the CBPI stack for FTQC architectures. 
Detailed methodologies for calculating the stack, similar to those explored for CPI in classical computing contexts~\cite{eyerman2017multi}, are left as future work. 
In this paper, we calculate the CBPI stack in the sequence of ``Base'', ``Magic'', ``Path'', and ``Decoding'' to simplify the analysis.

\section{Bypass architecture with 2.5D qubit layout\label{sec:proposal}}
\subsection{Overview\label{subsec:2.5d_arch}}

\begin{figure*}[tb]
    \centering
    \includegraphics[width=\linewidth]{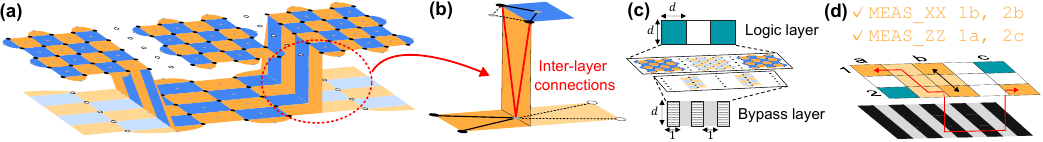}
    \caption{
    (a) A 3D view of the Bypass architecture.  
    (b) Inter-layer stabilizer in the circled area of (a) and its 3D qubits placement. 
    (c) Cell-level view of the Bypass architecture. 
    (d) Executing \measxx and \measzz operations simultaneously using Bypass layer.}
    \label{fig:2.5d_qubit_plane}
\end{figure*}

To achieve high-performance and scalable FTQC, we propose the Bypass architecture, which features a 2.5D qubit layout composed of a regular layer and a sparse layer. 
In the Bypass architecture, we call the regular (sparse) layer the \textit{Logic layer} (\textit{Bypass layer}). 
Figure~\ref{fig:2.5d_qubit_plane}\,(a) depicts an overview of the Bypass architecture during an LS operation.
Each cell in the Logic layer, which consists of $d \times d$ data qubits, has a corresponding ``SC fragment'' made up of $d \times 1$ data qubits in the Bypass layer, as shown in Fig.~\ref{fig:2.5d_qubit_plane}\,(c).
Each cell gap in the Logic layer has a corresponding fragment consisting of $d$ data and $2d$ ancillary qubits in the Bypass layer.

LS operations through the Bypass layer are performed using the SC fragments to make wide-rectangle-shaped intra-layer stabilizers, as shown in Fig~\ref{fig:2.5d_qubit_plane}\,(a).
By leveraging fragments, the Bypass layer performs longer LS operations with fewer physical qubits, mitigating decoding hazards. 
However, the path connections are restricted to directions orthogonal to the alignment of qubits within the Bypass layer.
In addition, the Bypass architecture provides multiple path options to reduce path conflicts.
These features can be likened to bypass surgery in medicine, where blood flow is redirected around blocked arteries to restore normal circulation.

\subsection{Benefits of Bypass architecture}
Because of two key features, which will be described in the following subsubsections, the Bypass architecture reduces path and decoding hazards and improves the LER of LS operations.
Note that the first feature is common to all the architectures with 3D-stacked qubit layers including the Bypass architecture, while the second is unique only to the Bypass architecture.

\subsubsection{Providing multiple LS paths options\label{subsubsec:3d_layer_construction}}

Despite the constraint that only vertical or horizontal paths can be connected via the Bypass layer, the Bypass architecture allows two LS paths that conflict on a 2D qubit plane to intersect.
As shown in Fig.~\ref{fig:2.5d_qubit_plane}\,(d), the Bypass architecture functions effectively in the same situation as Fig.~\ref{fig:hazards}\,(a) by executing the \measzz operation horizontally through the Bypass layer.
As a result, the Bypass architecture leverages its 3D-stacked structure to mitigate the path penalty.

In addition, in situations where roundabout LS paths are unavoidable on a 2D qubit plane, the Bypass architecture provides alternative, shorter LS path options through its 3D structure. 
As discussed in Sec.~\ref{subsubsec:decoding_process}, LS operations with shorter paths alleviate the demands on decoders.
Thus, the 3D structure of the Bypass architecture is also effective in mitigating the decoding penalty.

\subsubsection{Performing longer-path LS operations with fewer qubits\label{subsubsec:shorter_path}}

In the Bypass architecture, LS operations between distant data cells in the Logic layer can be executed with fewer physical qubits by utilizing SC fragments in the Bypass layer. 
Let $L$ denote the path length of an LS operation, defined as the number of cells allocated during the LS operation, including data cells. 
Let $L'$ represent the effective path length of an LS operation, defined as the number of data qubits involved in a given LS instruction divided by $d^2$.

Figure~\ref{fig:with_and_without_wiring_layer} compares two \measzz operations with and without the Bypass layer in a stabilizer-level view.
As shown in the figure, without a Bypass layer, executing a \measzz operation requires $d \times (Ld + L - 1) \sim O(Ld^2)$ data qubits, and the effective path length $L'$ is $O(L)$.
In contrast, employing a Bypass layer reduces the number of required data qubits to $d \times (2d + 2L - 3) \sim O(d^2 + Ld)$ and $L'$ to $O(L/d)$. 
Thus, the Bypass layer reduces the effective LS path length $L'$ by a factor of $d$, thereby mitigating the decoding penalty.

In addition, reducing the effective path length of LS operations by the Bypass architecture improves the fidelity of the entire FTQC program, as the LER of each LS operation is proportional to the number of qubits involved.
As a result, the Bypass architecture may enable the design of FTQC architectures with shorter code distances, requiring smaller quantum and classical hardware resources.

\begin{figure}[tb]
    \centering
    \includegraphics[width=\linewidth]{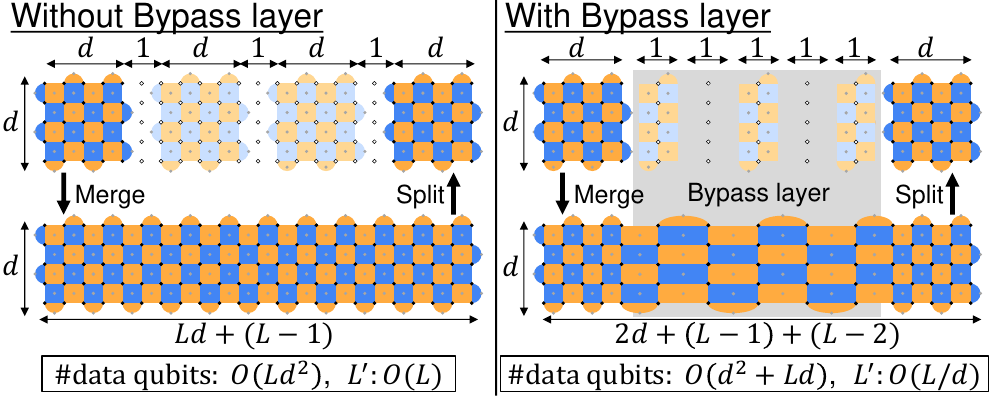}
    \caption{LS operation with and without Bypass layer.}
    \label{fig:with_and_without_wiring_layer}
\end{figure}

\subsection{Implementation\label{subsec:HW_implementation}}
\begin{figure*}[tb]
    \centering
    \includegraphics[width=\linewidth]{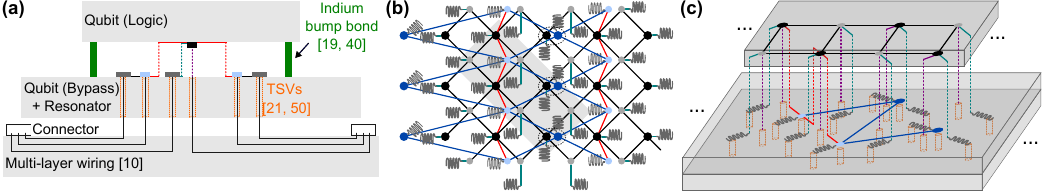}
    \caption{Implementation of the Bypass architecture. (a) Side view.
    (b) Top view. (c) Bird's-eye view of the shadowed part of (b). }
    \label{fig:HW_implementation}
\end{figure*}

Figure~\ref{fig:2.5d_qubit_plane}\,(b) shows the 3D arrangement of qubits that make up an inter-layer stabilizer between the Logic and Bypass layers.
As shown in the figure, we place a column of ancillary qubits in the Bypass layer vertically straight down from the leftmost column of data qubits of a cell in the Logic layer.
In addition to the intra-layer connections (represented as solid black lines in the figure) in the Bypass layer, the inter-layer CNOT operations (solid red lines) between the ancillary qubits in the Bypass layer and the data qubits in the Logic layer constitute the inter-layer stabilizer.

For the Bypass architecture, the qubits constituting inter-layer stabilizers require a 6-degree connection, including temporarily unused intra-layer connections (dashed black lines), while the others require a 4-degree connection.
Note that even for such qubits, the number of connections that a qubit uses simultaneously within one code cycle is kept at up to four.
Moreover, the proportion of qubits requiring a 6-degree connection is $\frac{2d}{d^2+(d+1)^2} \sim \frac{1}{d}$ per cell in the Logic layer, and $\frac{1}{3}$ in the Bypass layer.

In this paper, we focus mainly on implementation using superconducting qubits.
Several advanced fabrication technologies can be employed to achieve 3D qubit stacking, including through-silicon vias (TSVs)~\cite{yost2020solid,hazard2023characterization}, multi-layer wiring~\cite{dial2022eagle}, and flip-chip bonding~\cite{gold2021entanglement,smith2022scaling}. 
Stacking two qubit chips with comparable density, such as in a two-Logic-layer configuration, without compromising the density of the original qubit chip presents significant challenges. 
In contrast, a qubit chip for the Bypass layer can be introduced with minimal reduction in the density of the original Logic-layer qubit chip by leveraging the lower density of the Bypass layer.

The Bypass architecture can be implemented by positioning two qubit chips face-to-face and connecting them via flip-chip bonding (green bars), as shown in Fig.~\ref{fig:HW_implementation}\,(a). 
Here, all inter-chip couplings, including those for two-qubit gates for inter-layer stabilizers (red lines), are capacitive, as represented by dotted lines in the figure. 
The top view of Fig.~\ref{fig:HW_implementation}\,(b) shows qubits on the upper (lower) chip as black and gray (dark and light blue) circles and the resonators as dark gray marks.
Here, the lower chip is slightly shifted to the right to prevent undesirable interference between the qubits in the dotted circles in Fig.~\ref{fig:HW_implementation}\,(b) when they come closer, while this shift is not depicted in Figs.~\ref{fig:HW_implementation}\,(a) and (c) for simplicity.
As shown in Figs.~\ref{fig:HW_implementation}\,(b) and (c), the upper (lower) chip has a higher (lower) qubit density to form the Logic (Bypass) layer.
Control lines for the qubits on the upper chip (purple lines) are routed from the opposite side of the lower chip via TSVs (orange dotted cylinders).
The resonators connecting to the qubits on the top chip are located on the lower chip through inter-chip connections (teal lines), leveraging the low density of the lower chip.
As a result, the Bypass layer can be introduced with minimal reduction in the qubit density of the upper chip compared to that of a single-Logic-layer configuration.

Figure~\ref{fig:HW_implementation} illustrates one possible implementation of the Bypass architecture.
However, note that our architecture proposal and simulations are not limited to the specific implementation.

Stacking three or more qubit chips without compromising qubit density is challenging due to the significantly increased wiring complexity. 
Thus, the investigation of architectures with three or more layers is out of the scope of this paper and left as future work.

\subsection{Flexibility with different code distances}
\begin{figure}[tb]
    \centering
    \includegraphics[width=\linewidth]{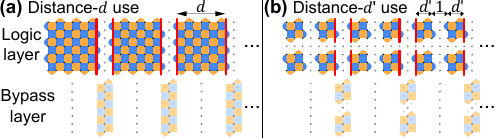}
    \caption{Use of Bypass with different code distances.}
    \label{fig:different_distances_usage}
\end{figure}

We need to predefine a certain code distance $d$ for each cell at the design phase; however, this design constraint does not significantly limit the flexibility of the code distance for each cell compared to the conventional 2D architecture. 
Figure~\ref{fig:different_distances_usage} shows an example where the Bypass architecture designed with $d = 7$ is used with cells having a code distance $d' = 3$.
Since the left and right boundaries of each cell are equivalent in terms of LS, as shown in Fig.~\ref{fig:different_distances_usage}\,(b), the cell can use the Bypass layer as long as its left or right boundary connects to the Bypass layer (red lines). 
In other words, when accommodating cells with a different code distance $d'$ than the design-phase distance $d$, the Bypass architecture can be employed by ensuring that either boundary of each cell aligns with the Bypass layer, albeit with a slight decrease in hardware utilization efficiency.

\subsection{Differences with long-range coupler}
The two benefits of the Bypass layer described in the previous subsection can also be achieved by connecting distant cells in the Logic layer using long-range couplers, such as assumed in Ref.~\cite{bravyi2024high}.
Our Bypass layer differs from this approach in the following points.

\noindent
\textbf{Programmability:} The Bypass layer dynamically determines connecting paths by selecting the stabilizers to activate, whereas the connections between cells through long-range couplers are fixed at the design phase.

\noindent
\textbf{Path length limit: }
The Bypass layer achieves LS paths of any length through local connections between adjacent qubits. 
For long-range couplers, as stated in the concluding remarks of Ref.~\cite{bravyi2024high}, connections that exceed a certain length determined by the frequency of qubits significantly increase PER.

Combined with the above differences, when considering each approach as a communication component, the long-range coupler acts merely as wiring, whereas the Bypass layer functions as a network switch and relay. 
The sparsely placed qubits in the Bypass layer programmatically determine the LS path and connect longer paths without significant impact on the operation error rate.

\section{LER evaluation\label{sec:QEC_performance}}
\subsection{Simulation setup and error model\label{subsec:QEC_sim_setup}}
We perform circuit-level stabilizer simulations using Stim~\cite{gidney2021stim} and PyMatching~\cite{higgott2023sparse} to estimate the LERs of \measzz operations with path length $L$ on the Bypass architecture. 
For simplicity, we describe $3d$-code-cycle protection of long SC patches, as shown at the bottom of Fig.~\ref{fig:with_and_without_wiring_layer}, using the Stim circuit format.
We assume a circuit-level noise model that applies a depolarizing noise channel after all physical gates and adds measurement errors.
For a given physical error rate (PER), we repetitively sample the error patterns, simulate the propagation of errors and the decoding procedure, and evaluate the probability of logical failure.

The \measzz operation using the Bypass involves inter-layer CNOT operations, as indicated by the red lines in Fig.~\ref{fig:2.5d_qubit_plane}\,(b). 
Although it significantly depends on the implementation, these inter-layer two-qubit gates may exhibit higher PER compared to intra-layer operations.
To evaluate the impact of inter-layer gates, we assume $\pinter$ as the PER of inter-layer operations, while all other operations have a PER of $p$.
We assume several values for $\pinter$ based on the infidelity of inter-chip communication with flip-chip bonding from Refs.~\cite{gold2021entanglement,smith2022scaling}. 
Optimistically, we assume $\pinter = p$. 
For a practical scenario, we assume $\pinter = 5p$, based on the ratio of inter- and intra-chip average PERs reported in Ref.~\cite{gold2021entanglement}.
Pessimistically, we assume $\pinter = 10p$, which is double the practical value.

\subsection{LER evaluation results}
\begin{figure*}[tb]
    \centering
    \includegraphics[width=\linewidth]{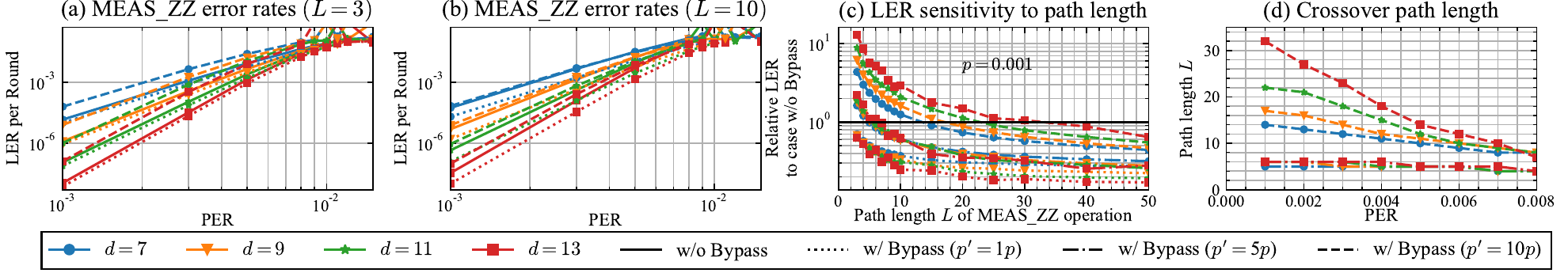}
    \caption{(a), (b) LER of \measzz operations with and without Bypass layer. 
    (c) Relative LER of bypassing \measzz relative to that without Bypass layer and its sensitivity to $L$. (d) Crossover path length $L$ where the Bypass layer improves LER.}
    \label{fig:logical_error_rate_and_relative}
\end{figure*}

Figures~\ref{fig:logical_error_rate_and_relative}\,(a) and (b) compare the LER of \measzz operations with and without a Bypass layer, as shown in Fig.~\ref{fig:with_and_without_wiring_layer}.
In the case of (a) $L=3$, the optimistic scenario with the Bypass layer (dotted lines) shows almost the same LER as that without the Bypass layer (solid lines).
By contrast, the pessimistic scenario (dashed lines) has a degraded performance due to the inter-layer operation compared to the solid lines.
In the case of (b) $L=10$, the optimistic scenario with the Bypass layer achieves lower LERs than those without Bypass.
For the pessimistic scenario with the Bypass layer, the performance degradation is moderate.

Figure~\ref{fig:logical_error_rate_and_relative}\,(c) shows the path length sensitivity of the LER of \measzz operations using a Bypass layer relative to that without a Bypass layer. 
Here, the PER $p$ is fixed to 0.001.
In the figure, values on the vertical axis below 1 indicate that using a Bypass layer results in a lower LER compared to that without a Bypass layer. 
The results show that the use of the Bypass layer is beneficial in terms of LER for longer \measzz operations even in the pessimistic scenario. 

Figure~\ref{fig:logical_error_rate_and_relative}\,(d) shows the minimum path length $L$ of \measzz operation where the use of the Bypass layer becomes advantageous in terms of LER, given a specific PER and code distance.
As the code distance decreases or the PER increases, the required $L$ for the Bypass layer to be advantageous also decreases.
In the evaluated parameter combinations, even in the pessimistic case, the Bypass layer is advantageous for \measzz with $L$ of 31 or more.

\section{FTQC performance evaluation\label{sec:CBPI_evaluation}}
\subsection{QPE as benchmark program\label{subsec:benchmark}}

\begin{table}[tb]
\tabcolsep=3.5pt
\centering
\caption{QPE programs used as benchmarks in this paper.\label{tab:benchmarks}}
\scriptsize
\begin{tabular}{|c|c|c|c|c|c|c|c|} \hline
Hamiltonian         & \multicolumn{4}{c|}{2D Fermi-Hubbard (FH)~\cite{yoshioka2022hunting}} & \textcolor{black}{3D Jellium\cite{kivlichan2020improved}}  &  \multicolumn{2}{c|}{\textcolor{black}{$H_4$~\cite{lee2021even}}}  \\ \hline
$N_{size}$          & 32     & 72     & 128    & 200          & \textcolor{black}{4}      &   \textcolor{black}{8}         &  \textcolor{black}{18}\\ 
\# data cells       & 268    & 340    & 428    & 532          & \textcolor{black}{316}      &   \textcolor{black}{244}       &  \textcolor{black}{382}\\
Total ops.          & 6428   & 14810  & 29000  & 50198        & \textcolor{black}{25448}      &   \textcolor{black}{6088}      &  \textcolor{black}{105314}\\ \hline
\textcolor{black}{1-qubit ops.} & \textcolor{black}{3460} & \textcolor{black}{8348} & \textcolor{black}{16772} & \textcolor{black}{29500} & \textcolor{black}{14516} & \textcolor{black}{3236} & \textcolor{black}{61790} \\ 
\textcolor{black}{2-qubit ops.} & \textcolor{black}{2968} & \textcolor{black}{6462} & \textcolor{black}{12228} & \textcolor{black}{20698} & \textcolor{black}{10932} & \textcolor{black}{2852} & \textcolor{black}{43524} \\ 
\textcolor{black}{Magic ops.} & \textcolor{black}{928} & \textcolor{black}{1760} & \textcolor{black}{3040} & \textcolor{black}{4768} & \textcolor{black}{2848} & \textcolor{black}{928} & \textcolor{black}{10016} \\ \hline
\end{tabular}

\end{table}

QPE is a quantum algorithm designed to estimate the ground state energy of a given Hamiltonian and is expected to demonstrate practical quantum advantages at an early stage, attracting attention across various research fields~\cite{reiher2017elucidating,yoshioka2022hunting,babbush2018encoding,kivlichan2020improved,lee2021even}.
In particular, Ref.~\cite{yoshioka2022hunting} focuses on QPE based on qubitization~\cite{low2019hamiltonian} for applications in condensed matter physics, including detailed resource estimation and the acceleration of its major subroutine, called the \texttt{SELECT} circuit. 
Figure~S11 of Ref.~\cite{yoshioka2022hunting} shows that over 80\% of QPE execution time is consumed by the \texttt{SELECT}. 
The \texttt{SELECT} circuit consists of a structure with multi-controlled gates, which can be decomposed with Toffoli gates and frequently appear in many quantum algorithms.
In addition, since quantum singular value transformation, an extension of qubitization, is known to include many other quantum algorithms, optimizing qubitization can lead to the acceleration of a wide range of quantum applications\cite{martyn2021grand}.
Thus, \texttt{SELECT} is well-suited as a benchmark because of its versatility for various FTQC applications and its significant contribution to execution time.

For the benchmark programs, we use \texttt{SELECT} circuits for Fermi-Hubbard (FH) and Jellium models, which are commonly used for condensed matter physics~\cite{yoshioka2022hunting,kivlichan2020improved}. 
For chemistry applications, we use \texttt{SELECT} circuits for $H_4$ molecule, where the Hydrogen atoms are located on $2 \times 2$ grid with a distance of 1.45 Angstrom, with the basis cc-pVDZ, similar to the resource estimation in Ref.~\cite{lee2021even}. 
Table~\ref{tab:benchmarks} summarizes these circuit characterizations. 
We mainly focus on the FH model with a problem size $N_{size}$ of 200, which represents the lattice size of the Hamiltonian. 
Note that QPE problems using the FH model with $N_{size} \geq 72$ are expected to demonstrate quantum advantages, as shown in Fig.~5 of Ref.~\cite{yoshioka2022hunting}. 
We utilize the parallelization technique of the \texttt{SELECT} circuit in Ref.~\cite{yoshioka2022hunting} with 16 threads for all simulations. 
Unless otherwise specified, the code distance $d$ is assumed to be 25, as in Ref.~\cite{yoshioka2022hunting}.

\subsection{Simulation setup\label{subsec:setup_LS_simulation}}
We evaluate the FTQC performance on the Bypass architecture using a cycle-accurate LS simulator.
Table~\ref{tab:parameters_for_LS_eval} summarizes the parameters used in the simulation. 
The underlined values in the table are used for the simulation in Fig.~\ref{fig:result_CBPI_stack}.
The bold values are the initial values for the sensitivity experiment in Fig.~\ref{fig:CPI_sensitivity_QPE} and also used for the evaluation in Figs.~\ref{fig:CBPI_preliminary} and \ref{fig:result_CBPI_various_application}.

\begin{table}[tb]
\caption{LS simulation parameters.\label{tab:parameters_for_LS_eval}}
\scriptsize
\centering
\begin{tabular}{|c|l|l|} \hline
Params.      & Description                  & Values                                   \\ \hline
Arch.        & Qubit layout                 & \underline{\textbf{1L-D}}, \underline{2L-DD}, \underline{2L-DP}, \underline{Bypass}   \\ 
$R_{data}$   & Data cell ratio (\%)       & 25, \underline{44}, \underline{\textbf{50}}                      \\ 
$n_{F}$      & Number of MSD circuits       & \underline{4}, \underline{8}, \underline{\textbf{12}}, 16                    \\
$TP_{dec}$   & Decoding throughput per cell & \underline{0.4}, \underline{\textbf{0.45}}, \underline{0.5}, 0.6, 0.8, 1.0   \\ \hline
\end{tabular}
\end{table}

Our simulation uses greedy instruction scheduling, where previously executed instructions occupy the required cells, and any executable instructions are always executed within a code beat. 
Each LS operation chooses the shortest available path.
Our greedy instruction scheduling avoids the situation where multiple operations attempt to occupy the same cells simultaneously.

%%% Qubit layouts
\begin{figure}[tb]
    \centering
    \includegraphics[width=\linewidth]{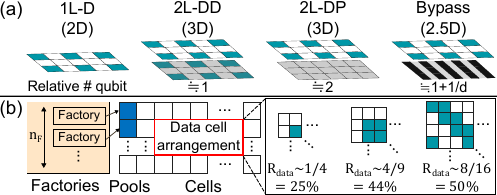}
    \caption{(a) Qubit layouts. (b) Floor plan and data cell arrangements with $R_{data} = 25\%$\cite{beverland2022surface}, $44\%$\cite{chamberland2022universal}, and $50\%$ (new).}
    \label{fig:overview_qubit_layouts}
\end{figure}

Our simulation compares the performance and hardware resources of four different qubit layouts: 1L, 2L-DD, 2L-DP, and Bypass, as shown in Fig.~\ref{fig:overview_qubit_layouts}\,(a).
The 1L-D layout consists of a single Logic layer with data cells. 
The 2L-DD and 2L-DP layouts consist of two Logic layers; 2L-DD allocates data cells in both layers with the same $R_{data}$, while 2L-DP allocates data cells in one layer with $R_{data}$ and uses the other as a pathway for LS operations.
The Bypass layout consists of one logic layer and one Bypass layer, enabling horizontal LS paths to be efficiently connected through it. 
For 2L-DD, 2L-DP, and Bypass layouts, the approximate number of total physical qubits relative to that of the 1L-D layout with the same $R_{data}$, is shown at the bottom of Fig.~\ref{fig:overview_qubit_layouts}\,(a).
Although we simulate the 2L-DD and 2L-DP layouts, note that their implementation feasibility is not addressed in Sec.~\ref{subsec:HW_implementation}.

Our Bypass architecture enables the execution of LS instructions with effectively shorter paths because of two key factors: 1) the availability of shorter path options via 3D-stacked layers, and 2) the sparsity of the Bypass layer, as explained in Sec.~\ref{subsubsec:3d_layer_construction} and \ref{subsubsec:shorter_path}, respectively.
By comparing 1L-D and 2L-DP (2L-DP and Bypass) layouts, the impact of Factor 1 (Factor 2) on $L'$ is estimated.
In addition, comparing 1L-D and 2L-DD estimates the impact of increasing the dimensionality of the qubit layout while maintaining the data cell density on $L'$.
Thus, the difference between the results of 1L-D and Bypass shows a comparison of the impact of the dimensional increase and the two Factors on $L'$.

%%% Data cell arrangement
In our simulation, as shown in the floor plan in Fig.~\ref{fig:overview_qubit_layouts}\,(b), the qubit plane is divided into three parts: Factories, Pools, and Cells.
The Factories part contains $n_F$ MSD circuits, each of which has a corresponding magic-state pool in the Pools part.
Data cells are allocated in the Cells part with variations in the data cell ratio $R_{data}$ from $25\%$ to $50\%$, as shown on the right side of Fig.~\ref{fig:overview_qubit_layouts}\,(b). 
In these data cell arrangements, the instructions described in Sec.~\ref{subsec:instruction_set} can always be executed on any data cell as long as hazards do not occur. 
We call such arrangements \textit{immediate operation (IO) capable}, and $R_{data}$ for IO-capable arrangements is at most 50\%, as proven in App.~\ref{sec:allocation}.
Thus, the arrangement with $R_{data}=50\%$, which we newly found in this paper, represents one of the densest IO-capable arrangements.

The positions of data cells remain fixed during computation, and each logical qubit is randomly assigned to a data cell at the start of the program.
To evaluate the impact of logical qubit assignment on performance, each simulation is repeated 1000 times with different random assignments.
The average result is used for evaluation, and the standard deviation is depicted as an error bar on a CBPI stack.

%%% wide allocation for bypass
Typically, data cells are arranged by repeating the pattern shown on the right side of Fig.~\ref{fig:overview_qubit_layouts}\,(b) an equal number of times vertically and horizontally to form a square-shaped Data cells area, which we refer to as the \textit{square arrangement}. 
For the Bypass layout, we consider a horizontally elongated arrangement of the Data cell area to leverage its characteristics, referred to as the \textit{wide arrangement}, although this slightly increases the total number of cells compared to the square arrangement.
In wide arrangements, the height of Data cell area is determined such that the average effective path length between any two data cells is minimized for a given number of data cells, as explained in Sec.~\ref{subsec:path_length_evaluation}. 
However, if this height is smaller than $n_F$, it is set to $n_F$.

\begin{figure}[tb]
    \centering
    \includegraphics[width=\linewidth]{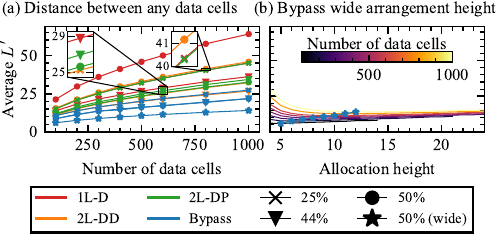}
    \caption{(a) Average $L'$ between any data cells. (b) Allocation height of Bypass with wide arrangement and average $L'$.}
    \label{fig:average_length}
\end{figure}

%%%%%%%%%%%%%%%
%%% Magic state factory
Based on the two-level 15-to-1 distillation protocol proposed in Ref.~\cite{litinski2019game}, we assume each MSD circuit generates a single magic state per 15 code beats. 
Magic states generated by each MSD circuit are stocked in the corresponding magic pool and consumed by non-Clifford operations via gate teleportation.
Note that the generation rate per MSD circuit may be improved by using more efficient strategies, such as the constructions proposed in Refs.~\cite{litinski2019magic,hirano2024leveraging,gidney2024magic}.

%%% decoding
During the LS simulation, the decoding tasks associated with each operation are considered as follows.
As shown in Fig.~\ref{fig:hazards}\,(d), whenever an instruction is executed, a decoding task associated with it is added to the decoding task queue.
Note that the decoding process is also required to protect data cells (\nopgate operations).
We assume that the difficulty of a decoding task for an operation is proportional to both the number of data qubits involved in the operation and its duration.
The decoder processes tasks sequentially from the top of the queue, provided that sufficient decoding capacity remains per code beat. 
We assume that the decoding capacity is proportional to the total number of qubits in the system. 
The parameter $TP_{dec}$ represents the decoding capacity per code beat per cell required for adequate online decoding.
Specifically, $TP_{dec} = 0.5 (1.0)$ represents an FTQC system with sufficient decoding resources to execute online decoding for half (all) of the cells. 
Note that the decoding task difficulty for a single cell in the Bypass layer is $1/d$ of that for a single cell in the Logic layer.
For simplicity, decoding tasks required for Factories are not considered.

\subsection{LS path length and program fidelity\label{subsec:path_length_evaluation}}

\begin{figure}[tb]
    \centering
    \includegraphics[width=\linewidth]{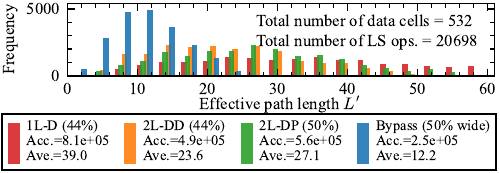}
    \caption{Histogram of $L'$ for QPE program with FH (200). }
    \label{fig:path_length_histogram}
\end{figure}

First, Fig.~\ref{fig:average_length}\,(a) shows the average effective distance $L'$ between any two data cells for each qubit layout and data cell arrangement across various numbers of data cells. 
The 1L-D and 2L-DD layouts exhibit similar trends: for $R_{data}$ ranging from 25\% to 44\%, the average $L'$ decreases as the total number of cells decreases with higher $R_{data}$. However, at $R_{data} = 50\%$, roundabout paths significantly increase the average $L'$.
In contrast, the 2L-DP layout utilizes its pathway layer to avoid roundabout paths, leading to a reduction in $L'$ as $R_{data}$ increases.
For sparse data cell arrangements with $R_{data} = 25\%$, where no roundabout paths exist even in the 1L-D layout, the average $L'$ is nearly identical to that of the 2L-DP layout (see the red and green crosses in the zoomed-in region on the right). 

The Bypass layout, compared to other layouts, effectively reduces the average $L'$, particularly in high-density data cell arrangements. 
Moreover, the advantage of the Bypass layout over the other three layouts scales with the number of data cells. 
This trend is particularly pronounced in the wide arrangement with $R_{data} = 50\%$ (blue line with stars).
Figure~\ref{fig:average_length}\,(b) shows the relationship between the height of the wide arrangement and the average $L'$ for the Bypass layout with $R_{data} = 50\%$.
For each data cell number, we choose the height that minimizes average $L'$ to generate the plot of Fig.~\ref{fig:average_length}\,(a).

Next, Fig.~\ref{fig:path_length_histogram} presents histograms of $L'$ for QPE with the FH model ($N_{size} = 200$) across the four qubit layouts.
We choose $R_{data}$ for each layout that minimizes average $L'$ in Fig.~\ref{fig:average_length}\,(a).
Other simulation parameters are set to the bolded values in Tab.~\ref{tab:parameters_for_LS_eval}.
As indicated in the legend of Fig.~\ref{fig:path_length_histogram}, the total $L'$ in the Bypass layout is approximately 30\% (half) of that in the 1L-D (2L-DD) layout.

The code distance $d$ is determined based on the PER $p$, the SC threshold $p_{th}$, and the LER $p_L$ required for the FTQC program, as described by the following formula: $p_L \approx \text{const} \times (p/p_{th})^{(d-1)/2}$.
As explained in Sec.~\ref{subsubsec:shorter_path}, the LER of the entire FTQC program is proportional to the total path length of all LS operations.
Thus, if the total $L'$ of the FTQC program is reduced by a factor of $1/X$, the $p_L$ of the overall program also decreases by a factor of $1/X$, allowing the code distance $d$ to be reduced according to the value of $p/p_{th}$.
As a result, assuming an SC threshold $p_{th}$ of 0.01 and a PER $p$ of approximately 0.003, the Bypass architecture enables a reduction in the code distance $d$ by 2 compared to the 1L-D, because of the total $L'$ results in Fig.~\ref{fig:path_length_histogram}. 
The impact of reducing $d$ on the FTQC performance and scalability is discussed in Sec.~\ref{subsec:performance_resource_tradeoff}.

\subsection{CBPI evaluation results}
\begin{figure}[tb]
    \centering
    \includegraphics[width=\linewidth]{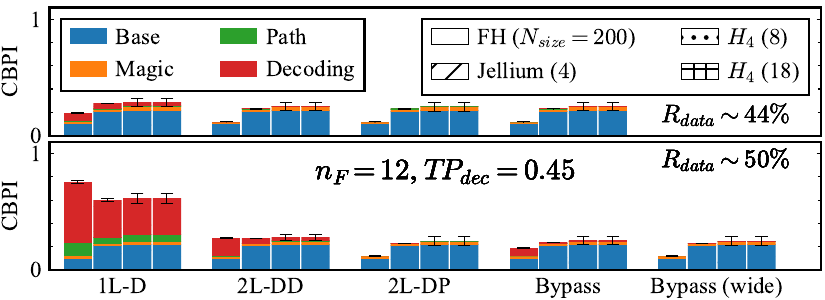}
    \caption{CBPI stacks for various QPEs.}
    \label{fig:result_CBPI_various_application}
\end{figure}
\begin{figure}[tb]
    \centering
    \includegraphics[width=\linewidth]{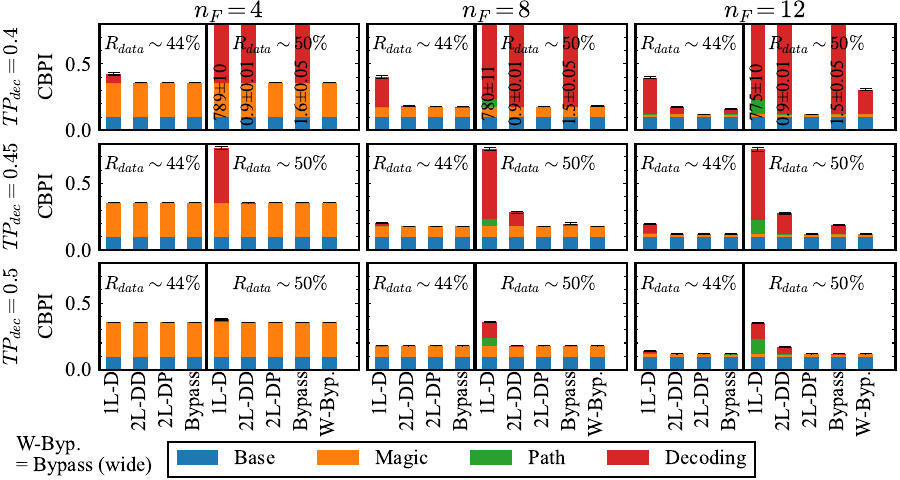}
    \caption{CBPI stacks of QPE FH model with $N_{size}=200$.}
    \label{fig:result_CBPI_stack}
\end{figure}

Figure~\ref{fig:result_CBPI_various_application} presents performance evaluation results of various QPE applications using the bolded parameters in Tab.~\ref{tab:parameters_for_LS_eval}, except for $R_{data}$.
The top and bottom plots show the results for $R_{data}=44\%$ and $50\%$, respectively.
For $R_{data}=44\%$, the path and decoding penalties slightly degrade the performance of the 1L-D layout but are completely mitigated by the other three layouts across all applications.
When $R_{data}=50\%$, the impact of these penalties increases in the 1L-D layout.
However, the 2L-DP layout and the Bypass with wide arrangement successfully eliminate them.

Since Fig.~\ref{fig:result_CBPI_various_application} shows a similar trend across all applications, we focus on the FH model for a more detailed evaluation. 
Figure~\ref{fig:result_CBPI_stack} presents the performance evaluation results for various parameter combinations of the underlined values in Tab.~\ref{tab:parameters_for_LS_eval}.
The results for parameters not underlined are omitted, as their selection had little to no impact on performance.
For example, increasing $TP_{dec}$ greater than 0.5 only eliminates the small impact of decoding hazards on CBPI in the cases with $R_{data} = 50\%$ in the figure.
Similarly, increasing $n_{F}$ from 12 to 16 only reduces the minor impact of magic hazards in any case.
For the cases with $R_{data} = 25\%$, decoding and path hazards have no impact on the CBPI.

For the cases of $TP_{dec} = 0.4$ (top row of the figure), increasing $R_{data}$ from 44\% to 50\% significantly reduces performance due to the decoding penalty, except in the 2L-DP layout and the Bypass with wide arrangement.
Note that the decoding resource is assumed to be proportional to the number of qubits in the layouts. 
The 2L-DP layout has nearly twice the decoding resources of the other three layouts for the same $TP_{dec}$ and $R_{data}$.
As the decoding resources increase (middle and bottom rows), the impact of the decoding penalty is significantly reduced in all cases.
At $TP_{dec} = 0.45$ (middle row), sufficient decoding throughput is achieved to mitigate the decoding penalty for the Bypass layout with $R_{data} = 50\%$ wide arrangement.
By contrast, the penalty persists in the 1L-D and 2L-DD layouts with the dense data cell arrangements.

The path penalty affects the performance only on the 1L-D layout with $R_{data}=50\%$ for the QPE programs. 
The other three layouts completely mitigate the hazards because of their 3D structures.

In cases with a limited number of MSDs, such as $n_F = 4$, the performance improvement with the Bypass layout is small, except when $TP_{dec} = 0.4$.
This occurs because the slow supply of magic states limits the number of LS instructions that can be executed simultaneously, reducing the chance of path and decoding hazards. 
As $n_F$ increases, these hazards degrade performance in the 1L-D and 2L-DD layouts, while the 2L-DP and Bypass mitigate them.

\begin{figure}[tb]
    \centering
    \includegraphics[width=\linewidth]{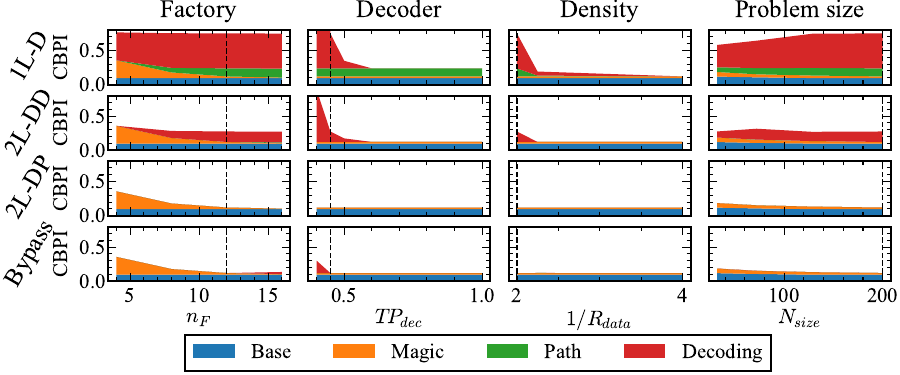}
    \caption{CBPI sensitivity results for QPE with FH models.}
    \label{fig:CPI_sensitivity_QPE}
\end{figure}

\begin{figure*}[tb]
    \centering
    \includegraphics[width=\linewidth]{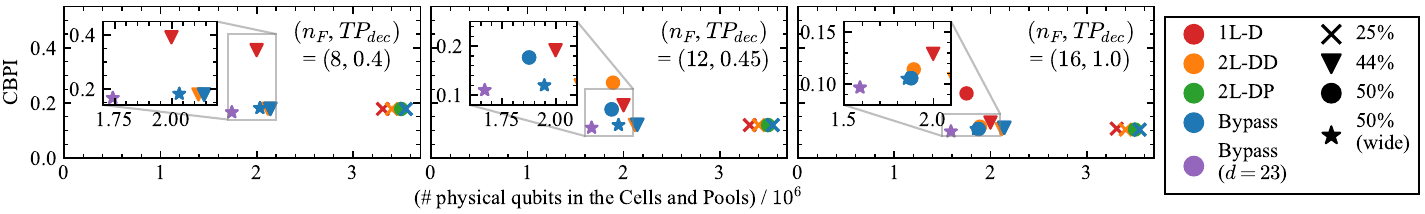}
    \caption{Tradeoff between CBPI and the number of physical qubits in the Cells and Pools parts.}
    \label{fig:CPI_qubits_tradeoff}
\end{figure*}

In all cases, the error bars, which represent the impact of logical qubit assignment on performance, in the Bypass layout with a wide arrangement are smaller than those for the other layouts, suggesting that the Bypass layout may reduce the burden of placement optimization in the LS compilation process, thereby simplifying the overall workflow and improving the efficiency of LS-based FTQC.

Next, Fig.~\ref{fig:CPI_sensitivity_QPE} summarizes the CBPI sensitivity to the number of factories, decoder resources, data cell arrangements, and problem size of QPE using various combinations of the values in Tab.~\ref{tab:parameters_for_LS_eval}.
One parameter is varied while the others are fixed to the bolded values in the table, represented as dashed lines in the figure.
In each plot, the changes from left to right on the graph represent design modifications that increase the required hardware resources.
Note that the Bypass layout with $R_{data} = 50\%$ uses wide arrangements, while the others use square arrangements.

The Factory column shows that increasing $n_{F}$ may not improve FTQC performance due to the impact of path and decoding penalties in the 1L-D and 2L-DD layouts. 
For $n_F = 16$ in the Bypass layout, the decoding penalty slightly increases because the height of the wide arrangement is set to $n_F$, reducing the number of horizontally aligned cells in the Cells part and limiting the opportunity to leverage the characteristics of the Bypass layer.
The Decoder and Density columns suggest that increasing classical and quantum resources can enhance FTQC performance when the magic state generation rate is sufficient for applications.
Meanwhile, the Bypass layout achieves high performance even with limited hardware resources.
In addition, the Problem size column indicates that this trend remains consistent regardless of the problem size.

\subsection{Tradeoff between performance and resource\label{subsec:performance_resource_tradeoff}}

This subsection presents the tradeoff between classical and quantum hardware resources and FTQC performance. 
Each graph in Fig.~\ref{fig:CPI_qubits_tradeoff} illustrates the various layouts, with the number of physical qubits in the Cells and Pools on the horizontal axis and CBPI on the vertical axis.
Note that the simulation assumes decoding resources to be proportional to the number of physical qubits.
Thus, the horizontal axis of each graph represents both the quantum hardware resources and the classical decoding resources.

As representative configurations, three cases are selected: limited resources ($n_F = 8,\ TP_{dec} = 0.4$), moderate resources ($12, 0.45$), and rich resources ($16, 1.0$), as shown in Fig.~\ref{fig:CPI_qubits_tradeoff}.
Particularly, we focus on comparing the 1L-D layout with $R_{data} = 44\%$ data cell arrangement and the Bypass layout with $R_{data} = 50\%$ wide arrangement. 
These are referred to as the base and proposed layouts, represented by the red triangles and blue stars in the figure, respectively.

In the limited resources case, the proposed layout achieves a 2.16$\times$ speedup with only 1.7\% additional hardware resources compared to the base layout.
In the moderate (rich) resources case, the proposed layout breaks the tradeoff between resources and performance, achieving both a 1.60$\times$ (1.23$\times$) speedup and a 2.6\% (7.2\%) reduction in hardware resources compared to the base.
In addition, the proposed layout achieves either higher performance, reduced hardware resources, or both, compared to any other layouts with 3D-stacked layers (orange and green markers) in all cases.

As discussed in Sec.~\ref{subsec:path_length_evaluation}, the Bypass architecture contributes to reducing the code distance by shortening $L'$ for LS operations.
The hardware resources and performance of the proposed layout, with the code distance reduced to 23, are represented as the purple stars in Fig.~\ref{fig:CPI_qubits_tradeoff}.
For the limited, moderate, and rich cases, the proposed layout with $d=23$ achieves 2.35$\times$, 1.73$\times$, and 1.33$\times$ speedup and 13\%, 17\%, and 21\% reduction in hardware resources compared to the base, respectively, demonstrating the further potential of the Bypass architecture for high-performance and scalable FTQC.

\section{Related work}
\noindent   
\textbf{FTQC resource estimation:} The depth of non-Clifford gates, such as $T$-gates, has often been used as a metric for FTQC resource estimation because each application of $T$-gates requires a time-consuming procedure consisting of magic-state injection and distillation in typical FTQC schemes, including LS~\cite{gidney2021rsa,tannu2017taming,ding2018magic}.
While the $T$-count can capture the time-scaling of given programs when magic-state preparations are the most time-consuming factors in large-scale FTQC, this estimation omits several vital factors in time analysis. 
As a result, the actual execution time can vary by a few orders of magnitude by the FTQC architecture and compilation.

\noindent   
\textbf{LS Compilation schemes:} To achieve high-throughput LS operations, optimal compilation, including data cell arrangement and LS instruction scheduling, is important; however, finding optimal LS compilation is known to be NP-hard~\cite{herr2017optimization,molavi2023compilation}.
Thus, fast and near-optimal approximation strategies for scheduling are well-studied. 
One possible strategy is to map the compilation problems into well-known NP-hard instances; Lao \textit{et al.}~\cite{lao2018mapping} map the problems into the quadratic assignment problem, and Molavi \textit{et al.}~\cite{molavi2023compilation} into SAT problems. 
As another approach, Hamada \textit{et al.}~\cite{hamada2024efficient} proposed a new scheduling method leveraging a technique to split large LS instructions into several smaller ones, such as Bell state preparation and measurements, and executing a part of them in advance.

\noindent
\textbf{Novel qubit layouts for LS-based FTQC:} Regarding architectural approaches for high-performance LS-based computation, Viszlai~\textit{et al.}~\cite{viszlai2023architecture} proposed a 3D LS architecture with multiple effective qubit layers leveraging the degrees of freedom in manipulating neutral atoms.
Their architecture achieves a throughput improvement by reducing path conflicts during LS operations and implementing transversal CNOT gates. 
Duckering \textit{et al.}~\cite{duckering2020virtualized} proposed a novel FTQC architecture named virtualized logical qubits, which supports transversal logical-CNOT gates and achieves hardware minimization through the combination of transmon and cavity qubits. 
Note that these approaches do not conﬂict with our proposal, and they are expected to be even more efﬁcient when combined.

\section{Conclusion\label{sec:conclusion}}
In this paper, we proposed a performance evaluation methodology named CBPI stack that breaks down the impact of each hazard on LS-based FTQC.
Based on the bottleneck analysis with the CBPI stack, we proposed the Bypass architecture to achieve high-performance and scalable LS-based FTQC by suppressing LS path length. 
In our simulation, the Bypass architecture achieves both a 1.73$\times$ speedup and a 17\% reduction in classical and quantum hardware resources compared to the 2D baseline in the moderate resources case, demonstrating its potential for high-performance and scalable FTQC.

\section*{Acknowledgement}
This work was partly supported by MEXT Q-LEAP Grant Numbers JPMXS0120319794 and JPMXS0118068682, JST Moonshot R\&D, Grant Numbers JPMJMS2061, JPMJMS2067-20 and JPMJMS226C-107, JST CREST Grant Numbers JPMJCR23I4 and JPMJCR24I4, JSPS KAKENHI Grant Numbers 22H05000, 22K17868, and 24K02915, RIKEN Special Postdoctoral Researcher Program.


\appendix
%\subfile{chap/appendix_densest_allocation}

\section{Dense data cell arrangement\label{sec:allocation}}

This section introduces a novel data cell arrangement to achieve an LS-based FTQC architecture with minimal hardware resources.
For simplicity, this paper focuses on the data cell arrangement where the instructions in the instruction set described in Sec.~\ref{subsec:instruction_set} can always be executed on any data cell, as long as hazards do not occur.
We call such arrangements as \textit{Immediate operation (IO) capable}.
The necessary conditions for IO-capable data cell arrangements are as follows:
\begin{enumerate}
    \renewcommand{\theenumi}{(\Roman{enumi})}
    \renewcommand{\labelenumi}{\theenumi}
    \item Each data cell has at least one ancillary cell to its left and right, and one above and below. \label{cond1}
    \item Any two ancillary cells have a path connecting them through ancillary cells. \label{cond2}
\end{enumerate}

The arrangements in Fig.~\ref{fig:datacell_allocation}\,(a), including the novel one that asymptotically achieves an $R_{data}$ of $50\%$, satisfy the conditions above.
By contrast, the arrangement in Fig.~\ref{fig:datacell_allocation}\,(b)~\cite{lee2021even} violates condition~\ref{cond1}.
As a result, it may be necessary to execute data cell rotation or movement protocols, which consume several code beats~\cite{litinski2019game}, before performing certain instructions, leading to performance degradation.

Here, the $R_{data}$ for any IO-capable data cell arrangement is at most 50\%, as proven in Thm.~\ref{theo:r_data}.
Thus, the new arrangement is one of the densest arrangements.

\begin{theorem}
    $R_{data}$ of any IO-capable arrangements is at most $50\%$. \label{theo:r_data} %%% 
\end{theorem}
\begin{proof}
  Let $N_d$ and $N_a$ be the numbers of data and ancillary cells in a certain region of the qubit plane.
  Let $R_{data} = \frac{N_d}{N_a + N_d}$ be the data cell ratio of the region.
  Let $D_i$, $A_i$, and $E_i$ represent the numbers of data cells, ancillary cells, and cells outside the region in the four adjacent cells of the $i$-th ancillary cell in the region, respectively.

  From Condition~\ref{cond1}, we have:
  \begin{equation}
    \sum\nolimits_i^{N_a} D_i \geq 2N_d,\ \text{and }\sum\nolimits_i^{N_a} E_i \geq 4.  \label{eq:di_sum}
  \end{equation}
  
  From Condition~\ref{cond2}, we have:
  \begin{equation}
    \sum\nolimits_i^{N_a} A_i \geq 2N_a - 2.\label{eq:ei_sum}
  \end{equation}

Combining Eqs.~\eqref{eq:di_sum} and \eqref{eq:ei_sum}, we get: 
  \begin{align}
    4N_a &= \sum\nolimits_i^{N_a} (D_i + A_i + E_i) \geq 2N_d + 2N_a + 2 \nonumber\\ 
    &\Leftrightarrow N_a \geq N_d + 1 \Rightarrow R_{data}  < 1/2  \nonumber \qedhere
  \end{align}
\end{proof}

\begin{figure}[tb]
    \centering
    \includegraphics[width=\linewidth]{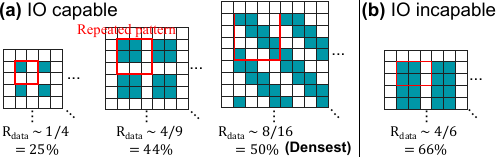}
    \caption{IO-capable and incapable arrangements with data cell density $R_{data} \sim 25\%$\cite{beverland2022surface}, $44\%$\cite{chamberland2022universal}, $50\%$ (densest), and $66\%$\cite{lee2021even}.}
    \label{fig:datacell_allocation}
\end{figure}

\bibliographystyle{ACM-Reference-Format}
\bibliography{refs}


\begin{thebibliography}{50}

%%% the Web address.

\ifx \showCODEN    \undefined \def \showCODEN     #1{\unskip}     \fi
\ifx \showDOI      \undefined \def \showDOI       #1{#1}\fi
\ifx \showISBNx    \undefined \def \showISBNx     #1{\unskip}     \fi
\ifx \showISBNxiii \undefined \def \showISBNxiii  #1{\unskip}     \fi
\ifx \showISSN     \undefined \def \showISSN      #1{\unskip}     \fi
\ifx \showLCCN     \undefined \def \showLCCN      #1{\unskip}     \fi
\ifx \shownote     \undefined \def \shownote      #1{#1}          \fi
\ifx \showarticletitle \undefined \def \showarticletitle #1{#1}   \fi
\ifx \showURL      \undefined \def \showURL       {\relax}        \fi
% invisible to TeX
\providecommand\bibfield[2]{#2}
\providecommand\bibinfo[2]{#2}
\providecommand\natexlab[1]{#1}
\providecommand\showeprint[2][]{arXiv:#2}

\bibitem[Babbush et~al\mbox{.}(2018)]%
        {babbush2018encoding}
\bibfield{author}{\bibinfo{person}{Ryan Babbush}, \bibinfo{person}{Craig
  Gidney}, \bibinfo{person}{Dominic~W Berry}, \bibinfo{person}{Nathan Wiebe},
  \bibinfo{person}{Jarrod McClean}, \bibinfo{person}{Alexandru Paler},
  \bibinfo{person}{Austin Fowler}, {and} \bibinfo{person}{Hartmut Neven}.}
  \bibinfo{year}{2018}\natexlab{}.
\newblock \showarticletitle{Encoding electronic spectra in quantum circuits
  with linear T complexity}.
\newblock \bibinfo{journal}{\emph{Physical Review X}} \bibinfo{volume}{8},
  \bibinfo{number}{4} (\bibinfo{year}{2018}), \bibinfo{pages}{041015}.
\newblock


\bibitem[Battistel et~al\mbox{.}(2023)]%
        {battistel2023real}
\bibfield{author}{\bibinfo{person}{Francesco Battistel},
  \bibinfo{person}{Christopher Chamberland}, \bibinfo{person}{Kauser Johar},
  \bibinfo{person}{Ramon~WJ Overwater}, \bibinfo{person}{Fabio Sebastiano},
  \bibinfo{person}{Luka Skoric}, \bibinfo{person}{Yosuke Ueno}, {and}
  \bibinfo{person}{Muhammad Usman}.} \bibinfo{year}{2023}\natexlab{}.
\newblock \showarticletitle{Real-time decoding for fault-tolerant quantum
  computing: progress, challenges and outlook}.
\newblock \bibinfo{journal}{\emph{Nano Futures}} \bibinfo{volume}{7},
  \bibinfo{number}{3} (\bibinfo{date}{aug} \bibinfo{year}{2023}),
  \bibinfo{pages}{032003}.
\newblock
\urldef\tempurl%
\url{https://doi.org/10.1088/2399-1984/aceba6}
\showDOI{\tempurl}


\bibitem[Beverland et~al\mbox{.}(2022a)]%
        {beverland2022surface}
\bibfield{author}{\bibinfo{person}{Michael Beverland}, \bibinfo{person}{Vadym
  Kliuchnikov}, {and} \bibinfo{person}{Eddie Schoute}.}
  \bibinfo{year}{2022}\natexlab{a}.
\newblock \showarticletitle{Surface code compilation via edge-disjoint paths}.
\newblock \bibinfo{journal}{\emph{PRX Quantum}} \bibinfo{volume}{3},
  \bibinfo{number}{2} (\bibinfo{year}{2022}), \bibinfo{pages}{020342}.
\newblock


\bibitem[Beverland et~al\mbox{.}(2022b)]%
        {beverland2022assessing}
\bibfield{author}{\bibinfo{person}{Michael~E Beverland},
  \bibinfo{person}{Prakash Murali}, \bibinfo{person}{Matthias Troyer},
  \bibinfo{person}{Krysta~M Svore}, \bibinfo{person}{Torsten Hoefler},
  \bibinfo{person}{Vadym Kliuchnikov}, \bibinfo{person}{Guang~Hao Low},
  \bibinfo{person}{Mathias Soeken}, \bibinfo{person}{Aarthi Sundaram}, {and}
  \bibinfo{person}{Alexander Vaschillo}.} \bibinfo{year}{2022}\natexlab{b}.
\newblock \showarticletitle{Assessing requirements to scale to practical
  quantum advantage (2022)}.
\newblock \bibinfo{journal}{\emph{arXiv preprint arXiv:2211.07629}}
  (\bibinfo{year}{2022}).
\newblock


\bibitem[Bravyi et~al\mbox{.}(2024)]%
        {bravyi2024high}
\bibfield{author}{\bibinfo{person}{Sergey Bravyi}, \bibinfo{person}{Andrew~W
  Cross}, \bibinfo{person}{Jay~M Gambetta}, \bibinfo{person}{Dmitri Maslov},
  \bibinfo{person}{Patrick Rall}, {and} \bibinfo{person}{Theodore~J Yoder}.}
  \bibinfo{year}{2024}\natexlab{}.
\newblock \showarticletitle{High-threshold and low-overhead fault-tolerant
  quantum memory}.
\newblock \bibinfo{journal}{\emph{Nature}} \bibinfo{volume}{627},
  \bibinfo{number}{8005} (\bibinfo{year}{2024}), \bibinfo{pages}{778--782}.
\newblock


\bibitem[Bravyi and Kitaev(2005)]%
        {bravyi_magic_state}
\bibfield{author}{\bibinfo{person}{Sergey Bravyi} {and} \bibinfo{person}{Alexei
  Kitaev}.} \bibinfo{year}{2005}\natexlab{}.
\newblock \showarticletitle{Universal quantum computation with ideal Clifford
  gates and noisy ancillas}.
\newblock \bibinfo{journal}{\emph{Phys. Rev. A}}  \bibinfo{volume}{71}
  (\bibinfo{date}{Feb} \bibinfo{year}{2005}), \bibinfo{pages}{022316}.
\newblock
Issue 2.
\urldef\tempurl%
\url{https://doi.org/10.1103/PhysRevA.71.022316}
\showDOI{\tempurl}


\bibitem[Bravyi and Kitaev(1998)]%
        {bravyi1998quantum}
\bibfield{author}{\bibinfo{person}{Sergey~B. Bravyi} {and}
  \bibinfo{person}{Alexei~Yu. Kitaev}.} \bibinfo{year}{1998}\natexlab{}.
\newblock \showarticletitle{Quantum codes on a lattice with boundary}.
\newblock \bibinfo{journal}{\emph{arXiv preprint quant-ph/9811052}}
  (\bibinfo{year}{1998}).
\newblock


\bibitem[Byun et~al\mbox{.}(2022)]%
        {byun2022xqsim}
\bibfield{author}{\bibinfo{person}{Ilkwon Byun}, \bibinfo{person}{Junpyo Kim},
  \bibinfo{person}{Dongmoon Min}, \bibinfo{person}{Ikki Nagaoka},
  \bibinfo{person}{Kosuke Fukumitsu}, \bibinfo{person}{Iori Ishikawa},
  \bibinfo{person}{Teruo Tanimoto}, \bibinfo{person}{Masamitsu Tanaka},
  \bibinfo{person}{Koji Inoue}, {and} \bibinfo{person}{Jangwoo Kim}.}
  \bibinfo{year}{2022}\natexlab{}.
\newblock \showarticletitle{{XQsim: modeling cross-technology control
  processors for 10+ K qubit quantum computers}}. In
  \bibinfo{booktitle}{\emph{Proceedings of the 49th Annual International
  Symposium on Computer Architecture}}. \bibinfo{pages}{366--382}.
\newblock


\bibitem[Chamberland and Campbell(2022)]%
        {chamberland2022universal}
\bibfield{author}{\bibinfo{person}{Christopher Chamberland} {and}
  \bibinfo{person}{Earl~T Campbell}.} \bibinfo{year}{2022}\natexlab{}.
\newblock \showarticletitle{Universal quantum computing with twist-free and
  temporally encoded lattice surgery}.
\newblock \bibinfo{journal}{\emph{PRX Quantum}} \bibinfo{volume}{3},
  \bibinfo{number}{1} (\bibinfo{year}{2022}), \bibinfo{pages}{010331}.
\newblock


\bibitem[Dial(2022)]%
        {dial2022eagle}
\bibfield{author}{\bibinfo{person}{Oliver Dial}.}
  \bibinfo{year}{2022}\natexlab{}.
\newblock \bibinfo{title}{Eagle’s quantum performance progress}.
\newblock
\newblock


\bibitem[Ding et~al\mbox{.}(2018)]%
        {ding2018magic}
\bibfield{author}{\bibinfo{person}{Yongshan Ding}, \bibinfo{person}{Adam
  Holmes}, \bibinfo{person}{Ali Javadi-Abhari}, \bibinfo{person}{Diana
  Franklin}, \bibinfo{person}{Margaret Martonosi}, {and}
  \bibinfo{person}{Frederic Chong}.} \bibinfo{year}{2018}\natexlab{}.
\newblock \showarticletitle{Magic-state functional units: Mapping and
  scheduling multi-level distillation circuits for fault-tolerant quantum
  architectures}. In \bibinfo{booktitle}{\emph{2018 51st Annual IEEE/ACM
  International Symposium on Microarchitecture (MICRO)}}. IEEE,
  \bibinfo{pages}{828--840}.
\newblock


\bibitem[Duckering et~al\mbox{.}(2020)]%
        {duckering2020virtualized}
\bibfield{author}{\bibinfo{person}{Casey Duckering},
  \bibinfo{person}{Jonathan~M Baker}, \bibinfo{person}{David~I Schuster}, {and}
  \bibinfo{person}{Frederic~T Chong}.} \bibinfo{year}{2020}\natexlab{}.
\newblock \showarticletitle{{Virtualized Logical Qubits}: {A 2.5 D}
  Architecture for Error-Corrected Quantum Computing}. In
  \bibinfo{booktitle}{\emph{2020 53rd Annual IEEE/ACM International Symposium
  on Microarchitecture}}. IEEE, \bibinfo{pages}{173--185}.
\newblock


\bibitem[Eyerman et~al\mbox{.}(2017)]%
        {eyerman2017multi}
\bibfield{author}{\bibinfo{person}{Stijn Eyerman}, \bibinfo{person}{Wim
  Heirman}, \bibinfo{person}{Kristof Du~Bois}, {and} \bibinfo{person}{Ibrahim
  Hur}.} \bibinfo{year}{2017}\natexlab{}.
\newblock \showarticletitle{Multi-stage CPI stacks}.
\newblock \bibinfo{journal}{\emph{IEEE Computer Architecture Letters}}
  \bibinfo{volume}{17}, \bibinfo{number}{1} (\bibinfo{year}{2017}),
  \bibinfo{pages}{55--58}.
\newblock


\bibitem[Fowler et~al\mbox{.}(2012)]%
        {fowler2012surface}
\bibfield{author}{\bibinfo{person}{Austin~G. Fowler}, \bibinfo{person}{Matteo
  Mariantoni}, \bibinfo{person}{John~M. Martinis}, {and}
  \bibinfo{person}{Andrew~N. Cleland}.} \bibinfo{year}{2012}\natexlab{}.
\newblock \showarticletitle{Surface codes: Towards practical large-scale
  quantum computation}.
\newblock \bibinfo{journal}{\emph{Phys. Rev. A}}  \bibinfo{volume}{86}
  (\bibinfo{date}{Sep} \bibinfo{year}{2012}), \bibinfo{pages}{032324}.
\newblock
Issue 3.
\urldef\tempurl%
\url{https://doi.org/10.1103/PhysRevA.86.032324}
\showDOI{\tempurl}


\bibitem[Gidney(2021)]%
        {gidney2021stim}
\bibfield{author}{\bibinfo{person}{Craig Gidney}.}
  \bibinfo{year}{2021}\natexlab{}.
\newblock \showarticletitle{Stim: a fast stabilizer circuit simulator}.
\newblock \bibinfo{journal}{\emph{{Quantum}}}  \bibinfo{volume}{5}
  (\bibinfo{date}{July} \bibinfo{year}{2021}), \bibinfo{pages}{497}.
\newblock
\showISSN{2521-327X}
\urldef\tempurl%
\url{https://doi.org/10.22331/q-2021-07-06-497}
\showDOI{\tempurl}


\bibitem[Gidney and Eker{\aa{}}(2021)]%
        {gidney2021rsa}
\bibfield{author}{\bibinfo{person}{Craig Gidney} {and} \bibinfo{person}{Martin
  Eker{\aa{}}}.} \bibinfo{year}{2021}\natexlab{}.
\newblock \showarticletitle{How to factor 2048 bit {RSA} integers in 8 hours
  using 20 million noisy qubits}.
\newblock \bibinfo{journal}{\emph{{Quantum}}}  \bibinfo{volume}{5}
  (\bibinfo{date}{April} \bibinfo{year}{2021}), \bibinfo{pages}{433}.
\newblock
\showISSN{2521-327X}
\urldef\tempurl%
\url{https://doi.org/10.22331/q-2021-04-15-433}
\showDOI{\tempurl}


\bibitem[Gidney and Fowler(2018)]%
        {gidney2019efficient}
\bibfield{author}{\bibinfo{person}{Craig Gidney} {and}
  \bibinfo{person}{Austin~G. Fowler}.} \bibinfo{year}{2018}\natexlab{}.
\newblock \showarticletitle{Efficient magic state factories with a
  catalyzed|CCZ⟩to2|T⟩transformation}.
\newblock \bibinfo{journal}{\emph{Quantum}} (\bibinfo{year}{2018}).
\newblock


\bibitem[Gidney et~al\mbox{.}(2024)]%
        {gidney2024magic}
\bibfield{author}{\bibinfo{person}{Craig Gidney}, \bibinfo{person}{Noah
  Shutty}, {and} \bibinfo{person}{Cody Jones}.}
  \bibinfo{year}{2024}\natexlab{}.
\newblock \showarticletitle{Magic state cultivation: growing T states as cheap
  as CNOT gates}.
\newblock \bibinfo{journal}{\emph{arXiv preprint arXiv:2409.17595}}
  (\bibinfo{year}{2024}).
\newblock


\bibitem[Gold et~al\mbox{.}(2021)]%
        {gold2021entanglement}
\bibfield{author}{\bibinfo{person}{Alysson Gold}, \bibinfo{person}{JP
  Paquette}, \bibinfo{person}{Anna Stockklauser}, \bibinfo{person}{Matthew~J.
  Reagor}, \bibinfo{person}{M.~Sohaib Alam}, \bibinfo{person}{Andrew~J.
  Bestwick}, \bibinfo{person}{Nicolas Didier}, \bibinfo{person}{Ani Nersisyan},
  \bibinfo{person}{Feyza~B. Oruç}, \bibinfo{person}{Armin Razavi},
  \bibinfo{person}{Ben Scharmann}, \bibinfo{person}{Eyob~A. Sete},
  \bibinfo{person}{Biswajit Sur}, \bibinfo{person}{Davide Venturelli},
  \bibinfo{person}{Cody~James Winkleblack}, \bibinfo{person}{Filip~A.
  Wudarski}, \bibinfo{person}{Mike Harburn}, {and} \bibinfo{person}{Chad~T.
  Rigetti}.} \bibinfo{year}{2021}\natexlab{}.
\newblock \showarticletitle{Entanglement across separate silicon dies in a
  modular superconducting qubit device}.
\newblock \bibinfo{journal}{\emph{npj Quantum Information}}
  \bibinfo{volume}{7}, \bibinfo{number}{1} (\bibinfo{year}{2021}),
  \bibinfo{pages}{142}.
\newblock


\bibitem[Hamada et~al\mbox{.}(2024)]%
        {hamada2024efficient}
\bibfield{author}{\bibinfo{person}{Kou Hamada}, \bibinfo{person}{Yasunari
  Suzuki}, {and} \bibinfo{person}{Yuuki Tokunaga}.}
  \bibinfo{year}{2024}\natexlab{}.
\newblock \showarticletitle{Efficient and high-performance routing of
  lattice-surgery paths on three-dimensional lattice}.
\newblock \bibinfo{journal}{\emph{arXiv preprint arXiv:2401.15829}}
  (\bibinfo{year}{2024}).
\newblock


\bibitem[Hazard et~al\mbox{.}(2023)]%
        {hazard2023characterization}
\bibfield{author}{\bibinfo{person}{Thomas~M. Hazard}, \bibinfo{person}{Wayne
  Woods}, \bibinfo{person}{Danna Rosenberg}, \bibinfo{person}{Rabindra Das},
  \bibinfo{person}{Cyrus~F. Hirjibehedin}, \bibinfo{person}{David~K. Kim},
  \bibinfo{person}{Jeffery Knecht}, \bibinfo{person}{Justin~L. Mallek},
  \bibinfo{person}{A. Melville}, \bibinfo{person}{Bethany~M. Niedzielski},
  \bibinfo{person}{Kyle Serniak}, \bibinfo{person}{Katrina~M. Sliwa},
  \bibinfo{person}{Donna Ruth-Yost}, \bibinfo{person}{Jonilyn~L. Yoder},
  \bibinfo{person}{William~D. Oliver}, {and} \bibinfo{person}{Mollie~E.
  Schwartz}.} \bibinfo{year}{2023}\natexlab{}.
\newblock \showarticletitle{Characterization of superconducting through-silicon
  vias as capacitive elements in quantum circuits}.
\newblock \bibinfo{journal}{\emph{Applied Physics Letters}}
  \bibinfo{volume}{123}, \bibinfo{number}{15} (\bibinfo{year}{2023}).
\newblock


\bibitem[Herr et~al\mbox{.}(2017)]%
        {herr2017optimization}
\bibfield{author}{\bibinfo{person}{Daniel Herr}, \bibinfo{person}{Franco Nori},
  {and} \bibinfo{person}{Simon~J Devitt}.} \bibinfo{year}{2017}\natexlab{}.
\newblock \showarticletitle{Optimization of lattice surgery is NP-hard}.
\newblock \bibinfo{journal}{\emph{Npj quantum information}}
  \bibinfo{volume}{3}, \bibinfo{number}{1} (\bibinfo{year}{2017}),
  \bibinfo{pages}{35}.
\newblock


\bibitem[Higgott and Gidney(2023)]%
        {higgott2023sparse}
\bibfield{author}{\bibinfo{person}{Oscar Higgott} {and} \bibinfo{person}{Craig
  Gidney}.} \bibinfo{year}{2023}\natexlab{}.
\newblock \showarticletitle{Sparse Blossom: correcting a million errors per
  core second with minimum-weight matching}.
\newblock \bibinfo{journal}{\emph{arXiv preprint arXiv:2303.15933}}
  (\bibinfo{year}{2023}).
\newblock


\bibitem[Hirano et~al\mbox{.}(2024)]%
        {hirano2024leveraging}
\bibfield{author}{\bibinfo{person}{Yutaka Hirano}, \bibinfo{person}{Tomohiro
  Itogawa}, {and} \bibinfo{person}{Keisuke Fujii}.}
  \bibinfo{year}{2024}\natexlab{}.
\newblock \showarticletitle{Leveraging Zero-Level Distillation to Generate
  High-Fidelity Magic States}.
\newblock \bibinfo{journal}{\emph{arXiv preprint arXiv:2404.09740}}
  (\bibinfo{year}{2024}).
\newblock


\bibitem[Holmes et~al\mbox{.}(2020)]%
        {holmes2020nisq}
\bibfield{author}{\bibinfo{person}{Adam Holmes}, \bibinfo{person}{Mohammad~Reza
  Jokar}, \bibinfo{person}{Ghasem Pasandi}, \bibinfo{person}{Yongshan Ding},
  \bibinfo{person}{Massoud Pedram}, {and} \bibinfo{person}{Frederic~T. Chong}.}
  \bibinfo{year}{2020}\natexlab{}.
\newblock \showarticletitle{{NISQ}+: Boosting Quantum Computing Power by
  Approximating Quantum Error Correction}. In
  \bibinfo{booktitle}{\emph{Proceedings of the ACM/IEEE 47th Annual
  International Symposium on Computer Architecture}}.
  \bibinfo{pages}{556–569}.
\newblock
\showISBNx{9781728146614}
\urldef\tempurl%
\url{https://doi.org/10.1109/ISCA45697.2020.00053}
\showDOI{\tempurl}


\bibitem[Horsman et~al\mbox{.}(2012)]%
        {horsman2012surface}
\bibfield{author}{\bibinfo{person}{Clare Horsman}, \bibinfo{person}{Austin~G
  Fowler}, \bibinfo{person}{Simon Devitt}, {and} \bibinfo{person}{Rodney
  Van~Meter}.} \bibinfo{year}{2012}\natexlab{}.
\newblock \showarticletitle{Surface code quantum computing by lattice surgery}.
\newblock \bibinfo{journal}{\emph{New Journal of Physics}}
  \bibinfo{volume}{14}, \bibinfo{number}{12} (\bibinfo{year}{2012}),
  \bibinfo{pages}{123011}.
\newblock


\bibitem[Itogawa et~al\mbox{.}(2024)]%
        {itogawa2024even}
\bibfield{author}{\bibinfo{person}{Tomohiro Itogawa}, \bibinfo{person}{Yugo
  Takada}, \bibinfo{person}{Yutaka Hirano}, {and} \bibinfo{person}{Keisuke
  Fujii}.} \bibinfo{year}{2024}\natexlab{}.
\newblock \showarticletitle{Even more efficient magic state distillation by
  zero-level distillation}.
\newblock \bibinfo{journal}{\emph{arXiv preprint arXiv:2403.03991}}
  (\bibinfo{year}{2024}).
\newblock


\bibitem[Kitaev(1997)]%
        {kitaev1997quantum}
\bibfield{author}{\bibinfo{person}{Alexei~Yu. Kitaev}.}
  \bibinfo{year}{1997}\natexlab{}.
\newblock \showarticletitle{Quantum computations: algorithms and error
  correction}.
\newblock \bibinfo{journal}{\emph{Russian Mathematical Surveys}}
  \bibinfo{volume}{52}, \bibinfo{number}{6} (\bibinfo{year}{1997}),
  \bibinfo{pages}{1191--1249}.
\newblock


\bibitem[Kivlichan et~al\mbox{.}(2020)]%
        {kivlichan2020improved}
\bibfield{author}{\bibinfo{person}{Ian~D. Kivlichan}, \bibinfo{person}{Craig
  Gidney}, \bibinfo{person}{Dominic~W. Berry}, \bibinfo{person}{Nathan Wiebe},
  \bibinfo{person}{Jarrod~R. McClean}, \bibinfo{person}{Wei Sun},
  \bibinfo{person}{Zhang Jiang}, \bibinfo{person}{Nicholas~C. Rubin},
  \bibinfo{person}{Austin~G. Fowler}, \bibinfo{person}{Al{\'a}n Aspuru-Guzik},
  \bibinfo{person}{Ryan Babbush}, {and} \bibinfo{person}{Hartmut Neven}.}
  \bibinfo{year}{2020}\natexlab{}.
\newblock \showarticletitle{Improved fault-tolerant quantum simulation of
  condensed-phase correlated electrons via trotterization}.
\newblock \bibinfo{journal}{\emph{Quantum}}  \bibinfo{volume}{4}
  (\bibinfo{year}{2020}), \bibinfo{pages}{296}.
\newblock


\bibitem[Knill(2005)]%
        {knill2005quantum}
\bibfield{author}{\bibinfo{person}{Emanuel Knill}.}
  \bibinfo{year}{2005}\natexlab{}.
\newblock \showarticletitle{Quantum computing with realistically noisy
  devices}.
\newblock \bibinfo{journal}{\emph{Nature}} \bibinfo{volume}{434},
  \bibinfo{number}{7029} (\bibinfo{year}{2005}), \bibinfo{pages}{39--44}.
\newblock


\bibitem[Lao et~al\mbox{.}(2018)]%
        {lao2018mapping}
\bibfield{author}{\bibinfo{person}{Lingling Lao}, \bibinfo{person}{Bas van
  Wee}, \bibinfo{person}{Imran Ashraf}, \bibinfo{person}{J van Someren},
  \bibinfo{person}{Nader Khammassi}, \bibinfo{person}{Koen Bertels}, {and}
  \bibinfo{person}{Carmen~G Almudever}.} \bibinfo{year}{2018}\natexlab{}.
\newblock \showarticletitle{Mapping of lattice surgery-based quantum circuits
  on surface code architectures}.
\newblock \bibinfo{journal}{\emph{Quantum Science and Technology}}
  \bibinfo{volume}{4}, \bibinfo{number}{1} (\bibinfo{year}{2018}),
  \bibinfo{pages}{015005}.
\newblock


\bibitem[Lee et~al\mbox{.}(2021)]%
        {lee2021even}
\bibfield{author}{\bibinfo{person}{Joonho Lee}, \bibinfo{person}{Dominic~W
  Berry}, \bibinfo{person}{Craig Gidney}, \bibinfo{person}{William~J Huggins},
  \bibinfo{person}{Jarrod~R McClean}, \bibinfo{person}{Nathan Wiebe}, {and}
  \bibinfo{person}{Ryan Babbush}.} \bibinfo{year}{2021}\natexlab{}.
\newblock \showarticletitle{Even more efficient quantum computations of
  chemistry through tensor hypercontraction}.
\newblock \bibinfo{journal}{\emph{PRX Quantum}} \bibinfo{volume}{2},
  \bibinfo{number}{3} (\bibinfo{year}{2021}), \bibinfo{pages}{030305}.
\newblock


\bibitem[Litinski(2019a)]%
        {litinski2019game}
\bibfield{author}{\bibinfo{person}{Daniel Litinski}.}
  \bibinfo{year}{2019}\natexlab{a}.
\newblock \showarticletitle{A game of surface codes: Large-scale quantum
  computing with lattice surgery}.
\newblock \bibinfo{journal}{\emph{Quantum}}  \bibinfo{volume}{3}
  (\bibinfo{year}{2019}), \bibinfo{pages}{128}.
\newblock


\bibitem[Litinski(2019b)]%
        {litinski2019magic}
\bibfield{author}{\bibinfo{person}{Daniel Litinski}.}
  \bibinfo{year}{2019}\natexlab{b}.
\newblock \showarticletitle{Magic state distillation: Not as costly as you
  think}.
\newblock \bibinfo{journal}{\emph{Quantum}}  \bibinfo{volume}{3}
  (\bibinfo{year}{2019}), \bibinfo{pages}{205}.
\newblock


\bibitem[Low and Chuang(2019)]%
        {low2019hamiltonian}
\bibfield{author}{\bibinfo{person}{Guang~Hao Low} {and}
  \bibinfo{person}{Isaac~L Chuang}.} \bibinfo{year}{2019}\natexlab{}.
\newblock \showarticletitle{Hamiltonian simulation by qubitization}.
\newblock \bibinfo{journal}{\emph{Quantum}}  \bibinfo{volume}{3}
  (\bibinfo{year}{2019}), \bibinfo{pages}{163}.
\newblock


\bibitem[Martyn et~al\mbox{.}(2021)]%
        {martyn2021grand}
\bibfield{author}{\bibinfo{person}{John~M Martyn}, \bibinfo{person}{Zane~M
  Rossi}, \bibinfo{person}{Andrew~K Tan}, {and} \bibinfo{person}{Isaac~L
  Chuang}.} \bibinfo{year}{2021}\natexlab{}.
\newblock \showarticletitle{Grand unification of quantum algorithms}.
\newblock \bibinfo{journal}{\emph{PRX quantum}} \bibinfo{volume}{2},
  \bibinfo{number}{4} (\bibinfo{year}{2021}), \bibinfo{pages}{040203}.
\newblock


\bibitem[Molavi et~al\mbox{.}(2023)]%
        {molavi2023compilation}
\bibfield{author}{\bibinfo{person}{Abtin Molavi}, \bibinfo{person}{Amanda Xu},
  \bibinfo{person}{Swamit Tannu}, {and} \bibinfo{person}{Aws Albarghouthi}.}
  \bibinfo{year}{2023}\natexlab{}.
\newblock \showarticletitle{Compilation for Surface Code Quantum Computers}.
\newblock \bibinfo{journal}{\emph{arXiv preprint arXiv:2311.18042}}
  (\bibinfo{year}{2023}).
\newblock


\bibitem[Reiher et~al\mbox{.}(2017)]%
        {reiher2017elucidating}
\bibfield{author}{\bibinfo{person}{Markus Reiher}, \bibinfo{person}{Nathan
  Wiebe}, \bibinfo{person}{Krysta~M Svore}, \bibinfo{person}{Dave Wecker},
  {and} \bibinfo{person}{Matthias Troyer}.} \bibinfo{year}{2017}\natexlab{}.
\newblock \showarticletitle{Elucidating reaction mechanisms on quantum
  computers}.
\newblock \bibinfo{journal}{\emph{Proceedings of the national academy of
  sciences}} \bibinfo{volume}{114}, \bibinfo{number}{29}
  (\bibinfo{year}{2017}), \bibinfo{pages}{7555--7560}.
\newblock


\bibitem[Skoric et~al\mbox{.}(2023)]%
        {skoric2023parallel}
\bibfield{author}{\bibinfo{person}{Luka Skoric}, \bibinfo{person}{Dan~E
  Browne}, \bibinfo{person}{Kenton~M Barnes}, \bibinfo{person}{Neil~I
  Gillespie}, {and} \bibinfo{person}{Earl~T Campbell}.}
  \bibinfo{year}{2023}\natexlab{}.
\newblock \showarticletitle{Parallel window decoding enables scalable fault
  tolerant quantum computation}.
\newblock \bibinfo{journal}{\emph{Nature Communications}} \bibinfo{volume}{14},
  \bibinfo{number}{1} (\bibinfo{year}{2023}), \bibinfo{pages}{7040}.
\newblock


\bibitem[Smith et~al\mbox{.}(2022)]%
        {smith2022scaling}
\bibfield{author}{\bibinfo{person}{Kaitlin~N Smith},
  \bibinfo{person}{Gokul~Subramanian Ravi}, \bibinfo{person}{Jonathan~M Baker},
  {and} \bibinfo{person}{Frederic~T Chong}.} \bibinfo{year}{2022}\natexlab{}.
\newblock \showarticletitle{Scaling superconducting quantum computers with
  chiplet architectures}. In \bibinfo{booktitle}{\emph{2022 55th IEEE/ACM
  International Symposium on Microarchitecture (MICRO)}}. IEEE,
  \bibinfo{pages}{1092--1109}.
\newblock


\bibitem[Tan et~al\mbox{.}(2024)]%
        {tan2024sat}
\bibfield{author}{\bibinfo{person}{Daniel~Bochen Tan},
  \bibinfo{person}{Murphy~Yuezhen Niu}, {and} \bibinfo{person}{Craig Gidney}.}
  \bibinfo{year}{2024}\natexlab{}.
\newblock \showarticletitle{{A SAT Scalpel for Lattice Surgery: Representation
  and Synthesis of Subroutines for Surface-Code Fault-Tolerant Quantum
  Computing}}.
\newblock \bibinfo{journal}{\emph{arXiv preprint arXiv:2404.18369}}
  (\bibinfo{year}{2024}).
\newblock


\bibitem[Tannu et~al\mbox{.}(2017a)]%
        {tannu2017cryogenic}
\bibfield{author}{\bibinfo{person}{Swamit~S. Tannu},
  \bibinfo{person}{Douglas~M. Carmean}, {and} \bibinfo{person}{Moinuddin~K.
  Qureshi}.} \bibinfo{year}{2017}\natexlab{a}.
\newblock \showarticletitle{{Cryogenic-DRAM} Based Memory System for Scalable
  Quantum Computers: A Feasibility Study}. In
  \bibinfo{booktitle}{\emph{Proceedings of the International Symposium on
  Memory Systems}}. \bibinfo{pages}{189–195}.
\newblock
\showISBNx{9781450353359}
\urldef\tempurl%
\url{https://doi.org/10.1145/3132402.3132436}
\showDOI{\tempurl}


\bibitem[Tannu et~al\mbox{.}(2017b)]%
        {tannu2017taming}
\bibfield{author}{\bibinfo{person}{Swamit~S. Tannu},
  \bibinfo{person}{Zachary~A. Myers}, \bibinfo{person}{Prashant~J. Nair},
  \bibinfo{person}{Douglas~M. Carmean}, {and} \bibinfo{person}{Moinuddin~K.
  Qureshi}.} \bibinfo{year}{2017}\natexlab{b}.
\newblock \showarticletitle{Taming the Instruction Bandwidth of Quantum
  Computers via Hardware-Managed Error Correction}. In
  \bibinfo{booktitle}{\emph{2017 50th Annual IEEE/ACM International Symposium
  on Microarchitecture}}. \bibinfo{pages}{679--691}.
\newblock


\bibitem[Terhal(2015)]%
        {terhal2015quantum}
\bibfield{author}{\bibinfo{person}{Barbara~M Terhal}.}
  \bibinfo{year}{2015}\natexlab{}.
\newblock \showarticletitle{Quantum error correction for quantum memories}.
\newblock \bibinfo{journal}{\emph{Reviews of Modern Physics}}
  \bibinfo{volume}{87}, \bibinfo{number}{2} (\bibinfo{year}{2015}),
  \bibinfo{pages}{307}.
\newblock


\bibitem[Ueno et~al\mbox{.}(2021)]%
        {ueno2021qecool}
\bibfield{author}{\bibinfo{person}{Yosuke Ueno}, \bibinfo{person}{Masaaki
  Kondo}, \bibinfo{person}{Masamitsu Tanaka}, \bibinfo{person}{Yasunari
  Suzuki}, {and} \bibinfo{person}{Yutaka Tabuchi}.}
  \bibinfo{year}{2021}\natexlab{}.
\newblock \showarticletitle{{QECOOL}: On-Line Quantum Error Correction with a
  Superconducting Decoder for Surface Code}. In \bibinfo{booktitle}{\emph{2021
  58th ACM/IEEE Design Automation Conference (DAC)}}.
  \bibinfo{pages}{451--456}.
\newblock
\urldef\tempurl%
\url{https://doi.org/10.1109/DAC18074.2021.9586326}
\showDOI{\tempurl}


\bibitem[Ueno et~al\mbox{.}(2022a)]%
        {ueno2022neo}
\bibfield{author}{\bibinfo{person}{Yosuke Ueno}, \bibinfo{person}{Masaaki
  Kondo}, \bibinfo{person}{Masamitsu Tanaka}, \bibinfo{person}{Yasunari
  Suzuki}, {and} \bibinfo{person}{Yutaka Tabuchi}.}
  \bibinfo{year}{2022}\natexlab{a}.
\newblock \showarticletitle{{NEO-QEC: Neural Network Enhanced Online
  Superconducting Decoder for Surface Codes}}.
\newblock \bibinfo{journal}{\emph{arXiv preprint arXiv:2208.05758}}
  (\bibinfo{year}{2022}).
\newblock


\bibitem[Ueno et~al\mbox{.}(2022b)]%
        {ueno2022qulatis}
\bibfield{author}{\bibinfo{person}{Yosuke Ueno}, \bibinfo{person}{Masaaki
  Kondo}, \bibinfo{person}{Masamitsu Tanaka}, \bibinfo{person}{Yasunari
  Suzuki}, {and} \bibinfo{person}{Yutaka Tabuchi}.}
  \bibinfo{year}{2022}\natexlab{b}.
\newblock \showarticletitle{{QULATIS: A Quantum Error Correction Methodology
  toward Lattice Surgery}}. In \bibinfo{booktitle}{\emph{2022 IEEE
  International Symposium on High-Performance Computer Architecture}}.
  \bibinfo{pages}{274--287}.
\newblock


\bibitem[Viszlai et~al\mbox{.}(2023)]%
        {viszlai2023architecture}
\bibfield{author}{\bibinfo{person}{Joshua Viszlai},
  \bibinfo{person}{Sophia~Fuhui Lin}, \bibinfo{person}{Siddharth Dangwal},
  \bibinfo{person}{Jonathan~M Baker}, {and} \bibinfo{person}{Frederic~T
  Chong}.} \bibinfo{year}{2023}\natexlab{}.
\newblock \showarticletitle{An Architecture for Improved Surface Code
  Connectivity in Neutral Atoms}.
\newblock \bibinfo{journal}{\emph{arXiv preprint arXiv:2309.13507}}
  (\bibinfo{year}{2023}).
\newblock


\bibitem[Yoshioka et~al\mbox{.}(2024)]%
        {yoshioka2022hunting}
\bibfield{author}{\bibinfo{person}{Nobuyuki Yoshioka},
  \bibinfo{person}{Tsuyoshi Okubo}, \bibinfo{person}{Yasunari Suzuki},
  \bibinfo{person}{Yuki Koizumi}, {and} \bibinfo{person}{Wataru Mizukami}.}
  \bibinfo{year}{2024}\natexlab{}.
\newblock \showarticletitle{Hunting for quantum-classical crossover in
  condensed matter problems}.
\newblock \bibinfo{journal}{\emph{npj Quantum Information}}
  \bibinfo{volume}{10}, \bibinfo{number}{1} (\bibinfo{year}{2024}),
  \bibinfo{pages}{45}.
\newblock


\bibitem[Yost et~al\mbox{.}(2020)]%
        {yost2020solid}
\bibfield{author}{\bibinfo{person}{D.~R.~W. Yost}, \bibinfo{person}{Mollie~E.
  Schwartz}, \bibinfo{person}{Justin~L. Mallek}, \bibinfo{person}{Danna
  Rosenberg}, \bibinfo{person}{Corey Stull}, \bibinfo{person}{Jonilyn~L.
  Yoder}, \bibinfo{person}{G. Calusine}, \bibinfo{person}{Matthew~T. Cook},
  \bibinfo{person}{Rabindra Das}, \bibinfo{person}{Alexandra~L. Day},
  \bibinfo{person}{Evan~B. Golden}, \bibinfo{person}{David~K. Kim},
  \bibinfo{person}{A. Melville}, \bibinfo{person}{Bethany~M. Niedzielski},
  \bibinfo{person}{Wayne Woods}, \bibinfo{person}{Andrew~J. Kerman}, {and}
  \bibinfo{person}{W.~D. Oliver}.} \bibinfo{year}{2020}\natexlab{}.
\newblock \showarticletitle{Solid-state qubits integrated with superconducting
  through-silicon vias}.
\newblock \bibinfo{journal}{\emph{npj Quantum Information}}
  \bibinfo{volume}{6}, \bibinfo{number}{1} (\bibinfo{year}{2020}),
  \bibinfo{pages}{59}.
\newblock


\end{thebibliography}


\end{document}